\documentclass{article}%
\usepackage{amsmath}
\usepackage{amsfonts}
\usepackage{amssymb}
\usepackage{graphicx}%
\providecommand{\U}[1]{\protect\rule{.1in}{.1in}}
\newtheorem{theorem}{Theorem}[section]

\newtheorem{corollary}[theorem]{Corollary}

\newtheorem{lemma}[theorem]{Lemma}

\newtheorem{proposition}[theorem]{Proposition}
\newtheorem{remark}[theorem]{Remark}

\newenvironment{proof}[1][Proof]{\noindent\textbf{#1.} }{\ \rule{0.5em}{0.5em}}
\makeatletter

\@addtoreset{equation}{section}
\makeatother
\begin{document}

\title{A new quantum version of $f$-divergence}
\author{Keiji Matsumoto\\Quantum Computation Group, National Institute of Informatics, \ \\2-1-2 Hitotsubashi, Chiyoda-ku, Tokyo 101-8430, \\e-mail:keiji@nii.ac.jp }
\maketitle

\section*{Abstract}

This paper proposes and studies new quantum version of $f$-divergences, a
class of convex functionals of a pair of probability distributions including
Kullback-Leibler divergence, Renyi-type relative entropy and so on. There are
several quantum versions so far, including the one by Petz
\cite{HiaiMosonyiPetzBeny}. We introduce another quantum version
($\mathrm{D}_{f}^{\max}$, below), defined as the solution to an optimization
problem, or the minimum classical $f$- divergence necessary to generate a
given pair of quantum states. It turns out to be the largest quantum
$f$-divergence. The closed formula of $\mathrm{D}_{f}^{\max}$ is given either
if $f$ is operator convex, or if one of the state is a pure state. Also,
concise representation of $\mathrm{D}_{f}^{\max}$ as a pointwise supremum of
linear functionals is given and used for the clarification of various
properties of the quality.

Using the closed formula of $\mathrm{D}_{f}^{\max}$, we show: Suppose $f$ is
operator convex. Then the\ maximum $f\,$- divergence of the probability
distributions of a measurement under the state $\rho$ and $\sigma$ is strictly
less than $\mathrm{D}_{f}^{\max}(\rho\Vert\sigma)$. This statement may seem
intuitively trivial, but when $f$ is not operator convex, this is not always
true. A counter example is $f(r)=\left\vert 1-r\right\vert $, which
corresponds to total variation distance.

We mostly work on finite dimensional Hilbert space, but some results are
extended to infinite dimensional case.

\section{Introduction}

\label{sec:introduction}

This paper proposes and studies a new quantum version of $f$-divergence:
\[
\mathrm{D}_{f}(p\Vert q):=\sum_{x}q(x)f\left(  \,p(x)/q(x)\,\right)  ,
\]
where $p$ and $q$ are probability distributions. Several important quantities
in information theory and statistics are in this class. For example,
$\mathrm{D}_{r\ln r}$ and $\mathrm{D}_{r^{\alpha}}$ correspond to
Kullback-Leibler divergence and Renyi-type relative entropy, respectively,
which are extensively used in asymptotic analysis of error probability of
decoding, hypothesis test, and so on.

Other $f$-divergences than these have at least one operational meaning. If $f$
is a convex function and satisfies some moderate conditions, $\mathrm{D}%
_{f}(p\Vert q)$ is the optimal gain of a certain Bayes decision problem: for
each $f$, there is a pair of functions $w_{1}$ and $w_{2}$ on decision space
representing a gain of decision $d$ with
\begin{equation}
\mathrm{D}_{f}(p\Vert q)=\sup_{d\left(  \cdot\right)  }\sum_{x}\left(
w_{1}\left(  d(x)\right)  p(x)+w_{2}\left(  d(x)\right)  q(x)\right)  .
\label{sup-w-0}%
\end{equation}
Conversely, for each $\left(  w_{1}\left(  \cdot\right)  ,w_{2}\left(
\cdot\right)  \right)  $, there is a convex function $f$ with this identity.
Also, by (\ref{sup-w-0}) and the celebrated randomization criterion
\cite{Strasser}, there is a Markov map which sends $\left(  p,q\right)  $ to
$\left(  p^{\prime},q^{\prime}\right)  $ iff $\mathrm{D}_{f}(p\Vert
q)\geq\mathrm{D}_{f}\left(  p^{\prime}\Vert q^{\prime}\right)  $ holds for any
convex function $f$ with above mentioned properties.

In quantum information theory, a series of works by Petz (see
\cite{HiaiMosonyiPetzBeny} and references therein) is most impressive, and his
version of quantum divergence have been widely studied and applied. Also,
recent development of theory of quantum Renyi entropy is significant.

In this paper, we introduce another quantum version, the maximal quantum
$f$\thinspace-\thinspace divergence $\mathrm{D}_{f}^{\max}(\rho\Vert\sigma)$.
This quantity is defined as the solution to the following optimization
problem: given a pair of quantum states $\left\{  \rho,\sigma\right\}  $,
consider a (completely) positive trace preserving map $\Gamma$ that sends
probability distributions $\left\{  p,q\right\}  $ to $\left\{  \rho
,\sigma\right\}  $. The triple $\left(  \Gamma,\left\{  p,q\right\}  \right)
$ (\textit{reverse test}, here after), is optimized to minimize $\mathrm{D}%
_{f}(p\Vert q)$, and this infimum is $\mathrm{D}_{f}^{\max}(\rho\Vert\sigma)$.
\ The name comes from the fact that $\mathrm{D}_{f}^{\max}$ is the largest of
the all possible quantum $f$\thinspace-\thinspace divergences.

Some historical remarks are in order. When $f$ is $r\ln r$ and $\sigma$ is
invertible,
\[
\mathrm{D}_{r\ln r}^{\max}(\rho\Vert\sigma)=\mathrm{tr}\,\rho\log\,\rho
^{1/2}\sigma^{-1}\rho^{1/2}.
\]
This RHS quantity had been studied by several authors from operator theoretic
point of view\thinspace\cite{Belavkin}\cite{HammersleyBelavkin}\cite{HiaiPetz}%
. Also, some authors had pointed out this quantity is path dependent
divergence\thinspace\cite{AmariNagaoka} of RLD quantum Fisher\thinspace metric
\cite{Petz}, which plays an important role in quantum statistical estimation
theory\thinspace\cite{Holevo}, along $e$- and $m$- geodesic connecting $\rho$
and $\sigma$\thinspace\cite{Hayashi}\cite{Jencova:03}\cite{Matsumoto:05}.
However, its characterization as the solution to the optimization problem and
the largest quantum version is first pointed out by the present author
\cite{Matsumoto:05}. In \cite{Matsumoto}, the present author studied
$\mathrm{D}_{r^{1/2}}^{\max}$ rather intensively, and briefly treated the case
when $f$ is operator monotone decreasing. Recently, based on an earlier
version of the present paper, \cite{HiaiMosonyi} studied some aspects of
$\mathrm{D}_{f}^{\max}$ .

Below, we summarize our main results. When $f$ is operator convex, a series of
rich results are available. First, we can write down the value of
$\mathrm{D}_{f}^{\max}$ \ and the operation achieving the minimum explicitly:
Suppose $\sigma$ and $\rho$ are%
\[
\rho=\left[
\begin{array}
[c]{cc}%
\rho_{11} & \rho_{12}\\
\rho_{21} & \rho_{22}%
\end{array}
\right]  ,\,\sigma=\left[
\begin{array}
[c]{cc}%
\sigma_{11} & 0\\
0 & 0
\end{array}
\right]  ,
\]
then
\begin{equation}
\mathrm{D}_{f}^{\max}\left(  \rho||\sigma\right)  =\mathrm{tr}\,\sigma
\,f\,\left(  \sigma^{-1/2}\tilde{\rho}\sigma^{-1/2}\right)  +\mathrm{tr}%
\,\left(  \rho-\tilde{\rho}\right)  \lim_{\varepsilon\downarrow0}\varepsilon
f\left(  1/\varepsilon\right)  . \label{Dmax-convex}%
\end{equation}
where $\tilde{\rho}:=\rho_{11}-\rho_{12}\left(  \rho_{22}\right)  ^{-1}%
\rho_{21}$. The first term of the RHS is trace of non-commutative
perspective\thinspace\cite{Ebadian}\cite{Effros} of $\tilde{\rho}$ and
$\sigma$. The operation achieves the minimum is obtained using spectral
decomposition of $\sigma^{-1/2}\tilde{\rho}\sigma^{-1/2}$, and the same
reverse test is optimal for all operator convex function $f$ 's . Uniqueness
of optimal operation modulo trivial redundancy is also shown.

Based on these analysis, we had shown, for example: Suppose $f$ is operator
convex. Then the\ maximum $f\,$- divergence of the probability distributions
of a measurement under the state $\rho$ and $\sigma$ is strictly less than
$\mathrm{D}_{f}^{\max}(\rho\Vert\sigma)$. Thus, once encoded into
non\thinspace-\thinspace commutative quantum states, some amount of classical
$f\,$-divergence is irrecoverably lost. (This statement may seem intuitively
trivial, but when $f$ is not operator convex, this is not always true.)

After the detailed analysis of the case of $f$ is operator convex, we study
the case where such an assumption is not true. One of the motivation is much
of the results in the former case generalizes.

First, when one of the states are a pure state, (\ref{Dmax-convex}) generalize
to all the convex functions, and the optimal reverse test is also the same,
and also unique. Also, $\mathrm{D}_{f}^{\max}$ is strictly larger than
measured $f$\thinspace-\thinspace divergence, unless two states commute.

Next, we analyzed $f(r)=\left\vert 1-r\right\vert $, since this corresponds to
total variation distance, which is quite often used in statistics, information
theory, and so on. Though we failed to obtain the closed formula, the
optimization problem is reduced to quite simple linear semidefinite program.
Using this, we had shown that (\ref{Dmax-convex}) is not true in this case,
and the optimal reverse test is no the same either.

In addition, when $\left\{  \rho,\sigma\right\}  $ satisfies some conditions,
it turns out
\begin{equation}
\mathrm{D}_{\left\vert 1-r\right\vert }^{\max}\left(  \rho||\sigma\right)
=\left\Vert \rho-\sigma\right\Vert _{1}.\label{D=TV2}%
\end{equation}
Since the RHS equals the measured total variation distance, this means the
total variation sometimes does not decrease by embedding into non\thinspace
-\thinspace commutative quantum states. The condition for (\ref{D=TV2}) is not
too restrictive: for example, if $\rho\sigma+\sigma\rho\geq0$, this identity
holds. In the qubit case, the necessary and sufficient condition for
(\ref{D=TV2}) is obtained, and fairly large area of Bloch sphere satisfies
(\ref{D=TV2}).

Besides from these case studies, we had shown the dual expression of
$\mathrm{D}_{f}^{\max}$,%
\begin{equation}
\mathrm{D}_{f}^{\max}(\rho\Vert\sigma)=\sup\left\{  \mathrm{tr}\left(  \,\rho
W_{1}+\sigma W_{2}\right)  \,;rW_{1}+W_{2}\leq f(r)\mathbf{1}\,,r\geq
0\right\}  . \label{Dmax-dual-2}%
\end{equation}
This shows that $\mathrm{D}_{f}^{\max}$ is the pointwise supremum of linear
functionals, thus it is lower\thinspace\ semicontinuous. Thus, $\mathrm{D}%
_{f}^{\max}$ behaves extremely nicely at the edge of the domain. In fact, if
$\left(  \rho_{\varepsilon},\sigma_{\varepsilon}\right)  $ is an arbitrary
line segment connecting $(\rho,\sigma)$ and an interior point of the domain,
$\lim_{\varepsilon\downarrow0}\mathrm{D}_{f}^{\max}\left(  \rho_{\varepsilon
}\Vert\sigma_{\varepsilon}\right)  =\mathrm{D}_{f}^{\max}(\rho\Vert\sigma)$
holds. (\ref{Dmax-dual-2}) is also valid, with certain restrictions, even when
the underlying Hilbert space is separable infinite dimensional space.

Except for the last subsection, we will work on a finite dimensional Hilbert
space $\mathcal{H}$. In most cases, the underlying Hilbert space is not
mentioned unless it is confusing. The space of trace class operators, and
bounded operators on $\mathcal{H}$ is denoted by $\mathcal{B}_{1}\left(
\mathcal{H}\right)  $, and $\mathcal{B}\left(  \mathcal{H}\right)  $,
respectively, and the space of their self\thinspace-\thinspace adjoint
elements are denoted by $\mathcal{B}_{1,sa}\left(  \mathcal{H}\right)  $, and
$\mathcal{B}_{sa}\left(  \mathcal{H}\right)  $. In most cases, specification
of underlying Hilbert space is dropped, thus $\mathcal{B}_{sa}$ in stead of
$\mathcal{B}_{sa}\left(  \mathcal{H}\right)  $, for example. When
$\dim\,\mathcal{H<\infty}$ (thus in most of the paper,) to denote the space of
all linear operators, we use $\mathcal{B}\left(  \mathcal{H}\right)  $.

For each operator $A$, $A^{-1}$ denotes its Moore-Penrose generalized inverse.
Also, for each positive operator $X$, denote by $\mathrm{supp}\,X$ the its
support, and by $\pi_{X}$ the projection onto $\mathrm{supp}\,X$. The
projection onto the space $\mathcal{K}$ is denoted by $\pi_{\mathcal{K}}$.
Orthogonal complement of the projector $\pi$ is denoted by $\pi^{\perp}$. In
this paper, in most part, probability distributions or positive measures are
defined on the finite set $\mathcal{X}$. These are easily identified with
commutative elements of $\mathcal{B}_{sa}\left(  \mathbb{C}^{\left\vert
\mathcal{X}\right\vert }\right)  $. Note the support of the measures $\mu$ is
also denoted by $\mathrm{supp}\,\mu$.

\section{Classical $f$-divergence}

\label{sec:classical}

This section explains the definition and known useful facts about classical
$f$-divergence, and convex analysis.

The definition of $\mathrm{D}_{f}$ in the introduction obviously cannot be
used when $q(x)=0$ for some $x$. Convex analysis supplies useful tools to cope
with such continuity issue. As in \cite{Rockafellar}, we suppose that $h$ is a
map from $\mathbb{R}^{n}$ $\ $to $\mathbb{R\cup}\left\{  \pm\infty\right\}  $.
Instead of saying that $h$ is not defined on a certain set, we say that
$h(r)=\infty$ on that set. The \textit{effective domain} of $h$, denoted by
$\mathrm{dom}\,h$ ,\ is the set of all $r$'s with $h(r)\mathbb{<\infty}$. $h$
is said to be convex iff its \textit{epigraph}, or the set $\mathrm{epi}%
\,h:=\left\{  \left(  r,\lambda\right)  ;\lambda\geq h(r)\right\}  $ is
convex. A convex function $h$ is \textit{proper} iff $h$ is nowhere $-\infty$
and not $\infty$ everywhere, and is \textit{lower semicontinuous} iff the set
$\left\{  r;\lambda\geq h(r)\right\}  $ is closed for any $\lambda$, or
equivalently, iff its epigraph is closed (Theorem\thinspace7.1 of
\cite{Rockafellar}) , or equivalently, $h(\lim_{k\rightarrow\infty}r_{k}%
)\leq\varliminf_{k\rightarrow\infty}h(r_{k})$.\ 

Given a convex function $h$, its \textit{closure} $\mathrm{cl}\,h$ is the
greatest lower semicontinuous (not necessarily finite) function majorized by
$h$. The name comes from the fact that $\mathrm{epi}\,(\mathrm{cl}%
\,h)=\mathrm{cl}\,(\mathrm{epi}\,h)$. $\mathrm{cl}\,h$ coincide with $h$
except perhaps at the relative boundary points of its effective domain. If $h$
is proper and convex, so is $\mathrm{cl}\,h$ (Theorem 7.4, \cite{Rockafellar}).

The following Proposition will be intensively used later.

\begin{proposition}
\label{prop:continuous}(Theorem 10.2,\cite{Rockafellar} ) If $h$ is lower
semicontinuous, proper and convex, it is continuous on any simplex in
$\mathrm{dom}\,h$.
\end{proposition}

From here, \ unless otherwise mentioned, $f$ , which is used to define
$f$-divergence, is supposed to satisfy the following condition.

\begin{description}
\item[(FC)] $f$ is a proper, lower semicontinuous, and convex function with
$\mathrm{dom}\,f\supset(0,\infty)$. Also, $f(0)=0$.
\end{description}

Now we are in the position to define the classical $f$-divergence
$\mathrm{D}_{f}$ between the positive measures $p$ and $q$ over the finite set
$\mathcal{X}$ . It is defined in the following manner, so that the function
$\left(  p,q\right)  \rightarrow\mathrm{D}_{f}\left(  p||q\right)  $ is lower
semicontinuous: Namely,
\[
\mathrm{D}_{f}(p\Vert q):=\sum_{x\in\mathcal{X}}g_{f}\left(  p(x),q(x)\right)
,
\]
where $g_{f}(s,t)$ is the closure of $tf\left(  \frac{s}{t}\right)  $ ( see
p.\thinspace35 and p.67 of \cite{Rockafellar} ),
\begin{equation}
g_{f}(s,t):=\left\{
\begin{array}
[c]{cc}%
tf\left(  s/t\right)  , & \text{if }s\in\mathrm{dom}\,f,t>0\\
\lim_{t\downarrow0}tf\left(  s/t\right)  , & \text{if }s\,\in\mathrm{dom}%
\,f,\,t=0,\\
0, & \text{if }s=t=0,\\
\infty, & \text{otherwise.}%
\end{array}
\right.  \label{def-g}%
\end{equation}

It is easy to check that \
\[
\mathrm{D}_{f}(p\Vert q)=\sum_{x\in\mathrm{supp}\,q}q(x)f\left(  \frac
{p(x)}{q(x)}\right)  +\sum_{x\in\mathcal{X}/\mathrm{supp}\,q}p(x)\lim
_{\varepsilon\downarrow0}\varepsilon\,f\left(  \frac{1}{\varepsilon}\right)
.
\]

\begin{remark}
Though $p$ and $q$ have to be probability for $\mathrm{D}_{f}$ to have
operational meanings, we extend the domain of $\mathrm{D}_{f}$ to pairs of
positive finite measures on finite set for the sake of mathematical convenience.
\end{remark}

Observe also $g_{f}$ is in addition positively homogeneous, or
\[
\forall a\geq0,\,\,g_{f}\left(  as,at\right)  =ag_{f}(s,t).\,
\]
\ Since it is positively homogeneous, proper, lower semicontinuous and convex,
by Corollary 13.5.1 of \cite{Rockafellar}, it is the pointwise supremum of
linear functions,%
\begin{equation}
g_{f}(s,t)=\sup_{\left(  w_{1},w_{2}\right)  \in\mathcal{W}_{f}}w_{1}s+w_{2}t,
\label{g=supW}%
\end{equation}
where the set $\mathcal{W}$ is convex and unbounded from below.

Therefore, \
\begin{equation}
\mathrm{D}_{f}(p\Vert q)=\sup\left\{  \sum_{x\in\mathcal{X}}w_{1}%
(x)p(x)+w_{2}(x)q(x);\left(  w_{1}(x),w_{2}(x)\right)  \in\mathcal{W}%
_{f}\text{ }\right\}  . \label{D=sup-w}%
\end{equation}
This in turn shows, by Corollary 13.5.1 of \cite{Rockafellar}, $\mathrm{D}%
_{f}$ is positively homogeneous, proper, lower semicontinuous and convex.

\begin{remark}
(\ref{D=sup-w}) indicates (\ref{sup-w-0}).To see this, use $\mathcal{W}_{f}$
as a decision space.
\end{remark}

If $f$ satisfies (FC), the function
\begin{equation}
\hat{f}(r):=g_{f}\left(  r,1\right)  \label{f-hat}%
\end{equation}
also satisfies (FC) and $g_{\hat{f}}\left(  t,s\right)  =g_{f}(s,t)$. This
identity implies%
\begin{equation}
\mathrm{D}_{f}\left(  p||q\right)  =\mathrm{D}_{\hat{f}}\left(  q||p\right)  .
\label{Df-Df}%
\end{equation}
Also,
\begin{equation}
\lim_{\varepsilon\downarrow0}\varepsilon f\left(  1/\varepsilon\right)
=\hat{f}(0),\,f(0)=\lim_{\varepsilon\downarrow0}\varepsilon\hat{f}\left(
1/\varepsilon\right)  . \label{f-hat-f}%
\end{equation}
Introduction of $\hat{f}$ often simplifies the argument, allowing to switch
the first and the second variables.

\section{Reverse test and maximal $f$- divergence}

\label{sec:max-f}

In this section we define maximal $f$- divergence $\mathrm{D}_{f}^{\max}$ as
the solution to an operationally defined minimization problem.

A \textit{reverse test} of a pair $\left\{  \rho,\sigma\right\}  $ of positive
definite operators is a triple $\left(  \Gamma,\left\{  p,q\right\}  \right)
$. Here, $\Gamma$ is a trace preserving positive linear map from positive
measures over some a finite set $\mathcal{X}$ (or commutative algebra with
dimension $\left\vert \mathcal{X}\right\vert $) to Hermitian operators, and
$p$ and $q$ are positive measures over $\mathcal{X}$, with
\[
\Gamma\left(  p\right)  =\rho,\,\Gamma\left(  q\right)  =\sigma.
\]
(Note $\Gamma$ is necessarily completely positive.)

For a function $f$ satisfying above (FC), we define \textit{maximal }%
$f$\textit{-divergence}
\begin{equation}
\mathrm{D}_{f}^{\max}(\rho\Vert\sigma)=\inf_{\left(  \Gamma,\left\{
p,q\right\}  \right)  }\mathrm{D}_{f}(p\Vert q), \label{Dmax-def}%
\end{equation}
where the infimum is taken over all the reverse tests. The name comes from the
fact that $\mathrm{D}_{f}^{\max}(\rho\Vert\sigma)$ is the largest quantum
version of $\mathrm{D}_{f}(p\Vert q)$; here, quantum version of $\mathrm{D}%
_{f}(p\Vert q)$ is any $\mathrm{D}_{f}^{Q}(\rho\Vert\sigma)$ such that

\begin{description}
\item[(D1)] $\mathrm{D}_{f}^{Q}\left(  \Lambda\left(  \rho\right)
\Vert\Lambda(\sigma)\right)  \leq\mathrm{D}_{f}^{Q}(\rho\Vert\sigma)$ holds
for any completely positive trace preserving (CPTP) map, any density operators
$\rho$,$\sigma$ on finite dimensional Hilbert spaces.

\item[(D2)] $\mathrm{D}_{f}^{Q}(p\Vert q)=\mathrm{D}_{f}(p\Vert q)$ for any
probability distributions $p$, $q$ over any finite sets.
\end{description}

Here $p$ is identified with $\sum_{x\in\mathcal{X}}p(x)\left\vert
e_{x}\right\rangle \left\langle e_{x}\right\vert $, for example, where
$\left\{  \left\vert e_{x}\right\rangle ;x\in\mathcal{X}\right\}  $ is a CONS.
Choice of a particular CONS is not important, since
\[
\mathrm{D}_{f}^{Q}\left(  U\rho U^{\dagger}\Vert U\sigma U^{\dagger}\right)
=\mathrm{D}_{f}^{Q}(\rho\Vert\sigma)
\]
for any unitary operator $U$ due to (D1).

We also consider the following stronger condition.

\begin{description}
\item[(D1')] $\mathrm{D}_{f}^{Q}\left(  \Lambda\left(  \rho\right)
||\Lambda(\sigma)\right)  \leq\mathrm{D}_{f}^{Q}\left(  \rho||\sigma\right)  $
holds for any trace preserving positive map $\Lambda$, any any density
operators $\rho$,$\sigma$, on finite dimensional Hilbert spaces.
\end{description}

\begin{lemma}
\label{lem:Dmax>DQ}If (FC) is satisfied, $\mathrm{D}_{f}^{\max}$ satisfies
above (D1), (D1') and (D2). Also, if a two point functional $\mathrm{D}%
_{f}^{Q}$ satisfies satisfies both of (D1) and (D2), or both of (D1') and
(D2),
\[
\mathrm{D}_{f}^{Q}(\rho\Vert\sigma)\leq\mathrm{D}_{f}^{\max}(\rho\Vert
\sigma).
\]

\end{lemma}

\begin{proof}
Let $\Lambda$ be a trace preserving positive map. Then,
\begin{align*}
&  \mathrm{D}_{f}^{\max}\left(  \Lambda\left(  \rho\right)  \Vert
\Lambda(\sigma)\right)  \\
&  =\inf_{\left(  \Gamma,\left\{  p,q\right\}  \right)  }\left\{
\mathrm{D}_{f}(p\Vert q);\left(  \Gamma,\left\{  p,q\right\}  \right)  \text{
: a reverse test of }\left\{  \Lambda\left(  \rho\right)  ,\Lambda
(\sigma)\right\}  \right\}  \\
&  \leq\inf_{\left(  \Gamma,\left\{  p,q\right\}  \right)  }\left\{
\mathrm{D}_{f}(p\Vert q);\Gamma=\Gamma^{\prime}\circ\Lambda,\left(
\Gamma^{\prime},\left\{  p,q\right\}  \right)  \text{ : a reverse test of
}\left\{  \rho,\sigma\right\}  \right\}  \\
&  =\mathrm{D}_{f}^{\max}(\rho\Vert\sigma).
\end{align*}
Hence, $\mathrm{D}_{f}^{\max}$ satisfies (D1'), and thus (D1) also. Also,
\begin{align*}
\mathrm{D}_{f}^{\max}(p\Vert q) &  =\inf\left\{  \mathrm{D}_{f}\left(
p^{\prime}\Vert q^{\prime}\right)  ;p=\Gamma\left(  p^{\prime}\right)
,q=\Gamma\left(  q^{\prime}\right)  ,\Gamma\text{: stochastic map}\right\}  \\
&  \geq\inf\left\{  \mathrm{D}_{f}\left(  \Gamma\left(  p^{\prime}\right)
\Vert\Gamma\left(  q^{\prime}\right)  \right)  ;p=\Gamma\left(  p^{\prime
}\right)  ,q=\Gamma\left(  q^{\prime}\right)  ,\Gamma\text{: stochastic
map}\right\}  \\
&  =\mathrm{D}_{f}(p\Vert q).
\end{align*}
Since the opposite inequality is trivial, we have $\mathrm{D}_{f}^{\max
}(p\Vert q)=\mathrm{D}_{f}(p\Vert q)$. Thus, $\mathrm{D}_{f}^{\max}$ satisfies (D2).

Suppose $\mathrm{D}_{f}^{Q}$ satisfies (D1) (or (D1')) and (D2), and let
$\left(  \Gamma,\left\{  p,q\right\}  \right)  $ be a reverse test of
$\left\{  \rho,\sigma\right\}  $. Then,
\[
\mathrm{D}_{f}^{Q}(\rho\Vert\sigma)=\mathrm{D}_{f}^{Q}\left(  \Gamma\left(
p\right)  \Vert\Gamma\left(  q\right)  \right)  \leq\mathrm{D}_{f}^{Q}(p\Vert
q)=\mathrm{D}_{f}(p\Vert q).
\]
Therefore, taking infimum over all the reverse tests of $\left\{  \rho
,\sigma\right\}  $, we have $\mathrm{D}_{f}^{Q}(\rho\Vert\sigma)\leq
\mathrm{D}_{f}^{\max}(\rho\Vert\sigma)$.
\end{proof}

\section{Representations of Reverse Tests}

\label{sec:reverse-test}

\subsection{Upper bound to the size of $\mathcal{X}$}

In defining reverse tests, we had assumed the cardinality of $\mathcal{X}$,
$\ $where $\left\{  p,q\right\}  $ are defined, is finite for mathematical
simplicity. But, this restriction is not essential as long as $\dim
\mathcal{H<\infty}$, since Caratheodory's theorem puts a natural upper bound
to the size of $\mathcal{X}$.  

Denote by $\delta_{x}$ the delta distribution at $x$, and define
\[
r_{x}:=p(x)/q(x),\,\,\text{if }x\in\mathrm{supp}\,q,
\]

Then
\begin{align*}
\sum_{x\in\mathcal{X}}q(x)\Gamma(\delta_{x}) &  =\sigma,\\
\sum_{x\in\mathcal{X}}q(x)r_{x}\Gamma(\delta_{x}) &  =\rho-\sum_{x:q(x)=0}%
p(x)\Gamma(\delta_{x}),\\
\sum_{x\in\mathcal{X}}q(x)f(r_{x}) &  =\mathrm{D}_{f}(p\Vert q)-\hat{f}%
(0)\sum_{x:q(x)=0}p(x).
\end{align*}
Since $\sum_{x\in\mathcal{X}}q(x)<\infty$, by Caratheodory's theorem, there is
a positive finite measure $\overline{q}$ such that $\sum_{x\in\mathcal{X}%
}\overline{q}(x)=\sum_{x\in\mathcal{X}}q(x)$,
\begin{align*}
\sum_{x\in\mathcal{X}}\overline{q}(x)\Gamma(\delta_{x}) &  =\sum
_{x\in\mathcal{X}}q(x)\Gamma(\delta_{x}),\\
\sum_{x\in\mathcal{X}}\overline{q}(x)r_{x}\Gamma(\delta_{x}) &  =\sum
_{x\in\mathcal{X}}q(x)r_{x}\Gamma(\delta_{x}),\\
\sum_{x\in\mathcal{X}}\overline{q}(x)f(r_{x}) &  =\sum_{x\in\mathcal{X}%
}q(x)f(r_{x}),
\end{align*}
and $\mathrm{supp}\,\overline{q}\subset\mathrm{supp}\,q$, $\left\vert
\mathrm{supp}\,\overline{q}\right\vert \leq\left(  \dim\mathcal{H}\right)
^{2}+\left(  \dim\mathcal{H}\right)  ^{2}+1+1$. Thus, defining $\overline
{\mathcal{X}}$, $\overline{p}$, and $\overline{\Gamma}$ by
\begin{align*}
\overline{\mathcal{X}} &  :=\mathrm{supp}\,\overline{q}\cup\{x_{0}\},\\
\overline{p}(x) &  :=\left\{
\begin{array}
[c]{cc}%
r_{x}q(x), & \text{if }x\in\mathrm{supp}\,\overline{q},\\
\sum_{x:q(x)=0}p(x), & \text{if }x=x_{0},
\end{array}
\right.  \,\\
\overline{\Gamma}(\delta_{x}) &  :=\left\{
\begin{array}
[c]{cc}%
\Gamma(\delta_{x}), & \text{if }x\in\mathrm{supp}\,\overline{q},\\
\sum_{x:q(x)=0}\Gamma(\delta_{x}), & \text{if }x=x_{0}%
\end{array}
\right.
\end{align*}
we have:

\begin{lemma}
\label{lem:caratheodory} To each given test $\left(  \Gamma,\left\{
p,q\right\}  \right)  $ of $\left\{  \rho,\sigma\right\}  $, there is a
reverse test $\left(  \overline{\Gamma},\left\{  \overline{p},\overline
{q}\right\}  \right)  $ such that (i) $\mathrm{D}_{f}(p\Vert q)=\mathrm{D}%
_{f}(\overline{p}\Vert\overline{q})$ and (ii) $\overline{p}$ and $\overline
{q}$ are defined over the set $\overline{\mathcal{X}}$ with $\left\vert
\overline{\mathcal{X}}\right\vert \leq2\left(  \dim\mathcal{H}\right)  ^{2}%
+3$. Also, $\overline{\mathcal{X}}\mathcal{=}\mathrm{supp}\,\overline{q}%
\cup\{x_{0}\}$. 
\end{lemma}

\subsection{A representation}

Without loss of generality, we suppose $\mathcal{X=}\mathrm{supp}%
\,q\cup\{x_{0}\}$, and define
\[
S_{x}:=\left\{
\begin{array}
[c]{cc}%
q(x)\Gamma(\delta_{x}), & \text{if }x\neq x_{0},,\\
p(x_{0})\Gamma(\delta_{0}), & \text{if }x=x_{0}.
\end{array}
\right.
\]
Then it should satisfy
\begin{equation}
\sum_{x\in\mathcal{X}\backslash\{x_{0}\}}r_{x}S_{x}+S_{x_{0}}=\rho
,\,\,\sum_{x\in\mathcal{X}\backslash\{x_{0}\}}S_{x}=\sigma.\label{S-r}%
\end{equation}
Conversely, to each such $\left\{  S_{x},r_{x};x\in\mathcal{X}\right\}  $,
there corresponds a reverse test with $\Gamma(\delta_{x})=\frac{1}%
{\mathrm{tr}\,S_{x}}S_{x}$ and%
\[
\left\{  p(x),\,q(x)\right\}  =\left\{
\begin{array}
[c]{cc}%
\left\{  r_{x}\mathrm{tr}\,S_{x},\mathrm{tr}\,S_{x}\right\}  , & \text{if
}x\neq x_{0},\\
\left\{  \mathrm{tr}\,S_{x_{0}},0\right\}  , & \text{if }x=x_{0}%
\text{\thinspace.}%
\end{array}
\right.  \,\,
\]
So $\left\{  S_{x},r_{x};x\in\mathcal{X}\right\}  $ is a bijective
representation of a reverse test. By this representation, $\mathrm{D}%
_{f}^{\max}$ is the solution to the optimization problem
\begin{equation}
\mathrm{D}_{f}^{\max}(\rho\Vert\sigma)=\inf\left\{  \sum_{x\in\mathcal{X}%
\backslash\{x_{0}\}}f(r_{x})\mathrm{tr}\,S_{x}+\hat{f}(0)\mathrm{tr}%
\,S_{x_{0}};S_{x}\text{ with (\ref{S-r})}\right\}  .\label{Dmax-S}%
\end{equation}

Observe this optimization can be done in the two stages; Fixing $S_{x_{0}}$,
or equivalently
\begin{equation}
\rho_{\ast}:=\rho-S_{x_{0}}=\sum_{x\in\mathcal{X}\backslash\{x_{0}\}}%
r_{x}\mathrm{tr}\,S_{x}\leq\rho,\label{rho*}%
\end{equation}
optimize $\left\{  S_{x};x\in\mathcal{X}\backslash\{x_{0}\}\right\}  $ to
minimize
\[
\sum_{x\in\mathcal{X}\backslash\{x_{0}\}}f(r_{x})\mathrm{tr}\,S_{x}%
=\sum_{x:\mathrm{supp}\,q(x)}q(x)f\left(  \frac{p(x)}{q(x)}\right)
=\mathrm{D}_{f}(p\Vert q),
\]
where $\tilde{p}$ is restriction of $p$ to $\mathrm{supp}\,q$. Since $\left(
\Gamma,\left\{  \tilde{p},q\right\}  \right)  $ is a reverse test of $\left\{
\rho_{\ast},\sigma\right\}  $, the minimum of $\mathrm{D}_{f}(\tilde{p}\Vert
q)$ equals $\mathrm{D}_{f}^{\max}\left(  \rho_{\ast}\Vert\sigma\right)  .$

After this is done, we optimize $\rho_{\ast}$to minimize
\[
\mathrm{D}_{f}^{\max}\left(  \rho_{\ast}\Vert\sigma\right)  +\hat
{f}(0)\mathrm{tr}\,S_{x_{0}}=\mathrm{D}_{f}^{\max}\left(  \rho_{\ast}%
\Vert\sigma\right)  +\hat{f}(0)\mathrm{tr}\,\left(  \rho-\rho_{\ast}\right)  .
\]
Note that $\mathrm{supp}\,\rho_{\ast}\subset\mathrm{supp}\,\sigma\,\ $holds by
(\ref{S-r}) and (\ref{rho*}), and $0\leq\rho_{\ast}\leq\rho$ by its definition
(\ref{rho*}). Thus,
\begin{equation}
\mathrm{D}_{f}^{\max}(\rho\Vert\sigma)=\inf\left\{  \mathrm{D}_{f}^{\max
}\left(  \rho_{\ast}\Vert\sigma\right)  +\hat{f}(0)\mathrm{tr}\,\left(
\rho-\rho_{\ast}\right)  ;0\leq\rho_{\ast}\leq\rho,\mathrm{supp}\,\rho_{\ast
}\subset\mathrm{supp}\,\sigma\right\}  \label{Dmax=0}%
\end{equation}

Here, introduce the operator \ {}%

\begin{equation}
\tilde{\rho}:=\rho_{11}-\rho_{12}\rho_{22}{}^{-1}\rho_{21}, \label{rho-tilde}%
\end{equation}
where $\ $%
\[
\rho=\left[
\begin{array}
[c]{cc}%
\rho_{11} & \rho_{12}\\
\rho_{21} & \rho_{22}%
\end{array}
\right]  ,\,\sigma=\left[
\begin{array}
[c]{cc}%
\sigma_{11} & 0\\
0 & 0
\end{array}
\right]  .
\]

\begin{lemma}
\label{lem:rho-tilde}Suppose $\rho_{\ast}\geq0$ is supported on $\mathrm{supp}%
\,\sigma$ and $\rho_{\ast}\leq\rho$. Then, $\tilde{\rho}\geq\rho_{\ast}.$
Also, $0\leq\tilde{\rho}\leq\rho$ and $\mathrm{supp}\,\tilde{\rho}$
$\subset\mathrm{supp}\,\sigma$.
\end{lemma}

\begin{proof}
By Proposition\thinspace\ref{prop:block-positive}, $\tilde{\rho}\geq0$. (in
fact, $\tilde{\rho}^{-1}=\pi_{\sigma}\rho^{-1}\pi_{\sigma}$.) \ $\tilde{\rho
}\leq\rho$ and $\mathrm{supp}\,\tilde{\rho}$ $\subset\mathrm{supp}\,\sigma$
are obvious by definition.

Since $\rho_{\ast}\leq\rho$,
\[
\left[
\begin{array}
[c]{cc}%
\rho_{11}-\rho_{\ast} & \rho_{12}\\
\rho_{21} & \rho_{22}%
\end{array}
\right]  \geq0,
\]
Therefore, by Proposition\thinspace\ref{prop:block-positive}, we should have
$\rho_{11}-\rho_{\ast}\geq\rho_{12}\rho_{22}^{-1}\rho_{21}$, or equivalently,
\[
\tilde{\rho}=\rho_{11}-\rho_{12}\rho_{22}^{-1}\rho_{21}\geq\rho_{\ast}.
\]

\end{proof}

By Lemma\thinspace\ref{lem:rho-tilde}, (\ref{Dmax=0}) can be rewritten as
follows:
\begin{align}
\mathrm{D}_{f}^{\max}(\rho\Vert\sigma)  &  =\inf\left\{  \mathrm{D}_{f}^{\max
}\left(  \rho_{\ast}\Vert\sigma\right)  +\hat{f}(0)\mathrm{tr}\,\left(
\rho-\rho_{\ast}\right)  \right\} \nonumber\\
&  =\inf\,\left\{  \mathrm{D}_{f}^{\max}\left(  \rho_{\ast}\Vert\sigma\right)
+\hat{f}(0)\mathrm{tr}\,\left(  \tilde{\rho}-\rho_{\ast}\right)  \right\}
+\hat{f}(0)\mathrm{tr}\,\left(  \rho-\tilde{\rho}\right) \nonumber\\
&  =\mathrm{D}_{f}^{\max}(\tilde{\rho}\Vert\sigma)+\hat{f}(0)\mathrm{tr}%
\,\left(  \rho-\tilde{\rho}\right)  , \label{Dmax=}%
\end{align}
where $\rho_{\ast}$ moves all the operators with $\rho\geq\rho_{\ast}\geq0$
and $\mathrm{supp}\,\rho_{\ast}\subset\mathrm{supp}\,\sigma$, or equivalently,
$\tilde{\rho}\geq\rho_{\ast}\geq0$ and $\mathrm{supp}\,\rho_{\ast}%
\subset\mathrm{supp}\,\sigma$.

\subsection{Representation and Radon-Nikodym derivative}

To list all the reverse tests, commutative Radon-Nikodym derivative is useful.
Given $\rho_{\ast}\geq0$ with $\mathrm{sppp\,}\rho_{\ast}\subset
\mathrm{sppp\,}\sigma$, the commutative Radon-Nikodym derivative with respect
to $\sigma$ is defined by
\begin{equation}
d\left(  \rho_{\ast},\sigma\right)  :=\sigma^{-1/2}\rho_{\ast}\sigma
^{-1/2}.\label{d-def}%
\end{equation}
Suppose $\rho_{\ast}\leq\rho$ and let $\left\{  M_{x}\right\}  $ be a
resolution of identity into positive operators with
\[
d\left(  \rho_{\ast},\sigma\right)  =\sum_{x\in\mathcal{X}\backslash\{x_{0}%
\}}r_{x}M_{x},\,\sum_{x\in\mathcal{X}}M_{x}=\mathbf{1}.\,
\]
Then
\[
S_{x}=\left\{
\begin{array}
[c]{cc}%
\sigma^{1/2}M_{x}\sigma^{1/2}, & \,\ \text{if }x\neq x_{0},\\
\frac{1}{\mathrm{tr}\left(  \rho-\rho_{\ast}\right)  \,}\left(  \rho
-\rho_{\ast}\right)  , & \text{if }x=x_{0}.
\end{array}
\right.
\]
Therefore,
\begin{equation}
\left\{  q(x),p(x)\right\}  =\left\{
\begin{array}
[c]{cc}%
\left\{  \mathrm{tr}\,\sigma M_{x},r_{x}q_{x}\right\}  , & \text{if }x\neq
x_{0},\\
\left\{  0,\mathrm{tr}\,\left(  \rho-\rho_{\ast}\right)  \right\}  , &
\text{if }x=x_{0},
\end{array}
\right.  \label{reverse}%
\end{equation}
and%
\begin{equation}
\Gamma(\delta_{x}):=\left\{
\begin{array}
[c]{cc}%
\frac{1}{q(x)}\sigma^{1/2}M_{x}\sigma^{1/2},\, & \text{if }x\neq x_{0},\\
\frac{1}{\mathrm{tr}\,\left(  \rho-\rho_{\ast}\right)  }\left(  \rho
-\rho_{\ast}\right)  , & \text{if }x=x_{0}.
\end{array}
\right.  \,\label{reverse-2}%
\end{equation}
Thus, $\left\{  M_{x},r_{x}\right\}  $ and $\rho_{\ast}$ specifies a reverse test.

When $M_{x}$'s are projectors and $\rho_{\ast}=\tilde{\rho}$, we say the
corresponding reverse test is \textit{minimal}.\thinspace The minimal reverse
test turns out to be optimal under certain natural conditions (namely, the
condition (F) in Section\thinspace\ref{sec:max-f}) on $f$.

\section{Properties of $\mathrm{D}_{f}^{\max}$}

\label{sec:property}

\begin{theorem}
\label{th:property}When $f$ satisfies (FC), $\mathrm{D}_{f}^{\max}$ has the
following properties.

(i) $\mathrm{D}_{f}^{\max}$ is jointly convex: if $\rho=\sum_{i}c_{i}\rho_{i}$
, $\sigma=\sum_{i}c_{i}\sigma_{i}$, $\sum_{i}c_{i}=1$ ($c_{i}\geq0$),
\begin{equation}
\mathrm{D}_{f}^{\max}(\rho\Vert\sigma)\leq\sum_{i}c_{i}\mathrm{D}_{f}^{\max
}\left(  \rho_{i}\Vert\sigma_{i}\right)  , \label{D-convex}%
\end{equation}

(ii) If $f(0)=0$ in addition, it is monotone decreasing in the second
argument:
\begin{equation}
\mathrm{D}_{f}^{\max}\left(  \rho\Vert X\right)  \leq\mathrm{D}_{f}^{\max
}(\rho\Vert\sigma),\,\,\,X\geq\sigma\label{Dmax-monotone}%
\end{equation}

(iii)$\mathrm{D}_{f}^{\max}$ is positively homogeneous.
\begin{equation}
\mathrm{D}_{f}^{\max}\left(  c\rho\Vert c\sigma\right)  =c\mathrm{D}_{f}%
^{\max}(\rho\Vert\sigma),\,c\geq0. \label{Dmax-homogeneous}%
\end{equation}
In particular,
\begin{equation}
\mathrm{D}_{f}^{\max}\left(  0\Vert0\right)  =0. \label{Dmax00}%
\end{equation}

(iv) Direct sum property:%
\begin{equation}
\mathrm{D}_{f}^{\max}\left(  \rho_{0}\oplus\rho_{1}\Vert\sigma_{0}\oplus
\rho_{1}\right)  =\mathrm{D}_{f}^{\max}\left(  \rho_{0}\Vert\sigma_{0}\right)
+\mathrm{D}_{f}^{\max}\left(  \rho_{1}\Vert\sigma_{1}\right)  ,
\label{Dmax-sum}%
\end{equation}
where $\,\rho_{i}$,$\,\sigma_{i}$ are supported on $\mathcal{H}_{i}$
($i=0,1$), and $\mathcal{H}_{0}\perp\mathcal{H}_{1}$.
\end{theorem}

\begin{proof}
(i): let $\left(  \Gamma_{i},\left\{  p_{i},q_{i}\right\}  \right)  $ be a
reverse tests of $\left\{  \rho_{i},\sigma_{i}\right\}  $, where $p_{i},q_{i}$
are positive measures over the finite set $\mathcal{X}_{i}$. Then `mixture' of
these reverse tests with probability $c_{i}$, compose a reverse test $\left(
\Gamma,\left\{  p,q\right\}  \right)  $\ of $\left\{  \rho,\sigma\right\}  $:
let $\mathcal{X}=\bigcup_{i}\mathcal{X}_{i}$, and define
\[
p_{i}(x):=c_{i}p_{0},q_{i}(x):=c_{i}q_{i}(x),\text{\thinspace}\Gamma
(\delta_{x})=\Gamma_{i}(\delta_{x}),\,\left(  x\in\mathcal{X}_{i}\right)  .
\]
Then,
\[
\mathrm{D}_{f}\left(  p\Vert q\right)  =\sum_{i}c_{i}\mathrm{D}_{f}\left(
p_{i}\Vert q_{i}\right)  .
\]
Therefore, minimizing over all the reverse tests of $\left\{  \rho
,\sigma\right\}  $, we obtain (\ref{D-convex}).

(ii): let $X^{\prime}:=X-\sigma\geq0$, and define $\mathcal{X}_{1}%
=\mathrm{supp}q\cup$ $\mathrm{supp}p$ and $\mathcal{X}_{2}\cap\mathcal{X}%
_{1}=\emptyset$. Then
\begin{align*}
&  \mathrm{D}_{f}^{\max}\left(  \rho\Vert X\right)  =\mathrm{D}_{f}^{\max
}\left(  \rho\Vert\sigma+X^{\prime}\right) \\
&  \leq\inf\left\{  \mathrm{D}_{f}\left(  p\Vert q+q^{\prime}\right)
\,;\Gamma\left(  p\right)  =\rho,\Gamma\left(  q\right)  =\sigma,\Gamma\left(
q^{\prime}\right)  =X,\,\mathrm{supp\,}q^{\prime}=\mathcal{X}_{2}\right\} \\
&  =\inf\left\{  \mathrm{D}_{f}(p\Vert q)\,;\Gamma\left(  p\right)
=\rho,\Gamma\left(  q\right)  =\sigma,\Gamma\left(  q^{\prime}\right)
=X,\,\mathrm{supp\,}q^{\prime}=\mathcal{X}_{2}\right\} \\
&  =\mathrm{D}_{f}^{\max}(\rho\Vert\sigma),
\end{align*}
where the identity in the third line is due to: since $\mathcal{X}_{2}%
\cap\mathcal{X}_{1}=\emptyset$ and $f(0)=0$,
\begin{align*}
\sum_{x\in\mathcal{X}_{1}\cup\mathcal{X}_{2}}g_{f}\left(  p(x),q(x)+q^{\prime
}(x)\right)   &  =\sum_{x\in\mathcal{X}_{1}}g_{f}\left(  p(x),q(x)\right)
+\sum_{x\in\mathcal{X}_{2}}g_{f}\left(  0,q^{\prime}(x)\right) \\
&  =\sum_{x\in\mathcal{X}_{1}}g_{f}\left(  p(x),q(x)\right)
\end{align*}

(iii): Let $c>0$. Then to each reverse test $\left(  \Gamma,\left\{
p,q\right\}  \right)  $ of $\left\{  \rho,\sigma\right\}  $, corresponds the
reverse test $\left(  \Gamma,\left\{  cp,cq\right\}  \right)  $ test of
$\left\{  c\rho,c\sigma\right\}  $, and vice versa. Hence, due to the fact
that $\mathrm{D}_{f}$ positively homogeneous, we have the identity. When
$c=0$, we only have to show the LHS is $0$. In fact, if $\left(
\Gamma,\left\{  p,q\right\}  \right)  $ is an arbitrary reverse test of
$\left\{  0,0\right\}  $, $\mathrm{supp}\,p$ and $\mathrm{supp}\,q$ are empty,
and $\mathrm{D}_{f}(p\Vert q)=0$. Thus $\mathrm{D}_{f}^{\max}\left(
0\Vert0\right)  =0\,$.

(iv): "$\leq$" is trivial. Thus, we show "$\geq$". Let $\left(  \Gamma
,\left\{  p,q\right\}  \right)  $ be a reverse test of $\left\{  \rho
,\sigma\right\}  =\left\{  \rho_{0}\oplus\rho_{1},\sigma_{0}\oplus\sigma
_{1}\right\}  $, and define
\begin{align*}
\Gamma_{i}(\delta_{x})  &  :=\frac{1}{\mathrm{tr}\,\pi_{\mathcal{H}_{i}}%
\Gamma(\delta_{x})}\pi_{\mathcal{H}_{i}}\Gamma(\delta_{x})\pi_{\mathcal{H}%
_{i}},\\
p_{i}(x)  &  :=p(x)\mathrm{tr}\,\pi_{\mathcal{H}_{i}}\Gamma(\delta
_{x}),\,\,q_{i}(x):=p(x)\mathrm{tr}\,\pi_{\mathcal{H}_{i}}\Gamma(\delta_{x}).
\end{align*}
Then $\left(  \Gamma_{i},\left\{  p_{i},q_{i}\right\}  \right)  $ is a reverse
test of $\left\{  \rho_{i},\sigma_{i}\right\}  $ ($i=0,1$). Also, since
$g_{f}$ is positively homogeneous,
\begin{align*}
&  \mathrm{D}_{f}\left(  p_{0}\Vert q_{0}\right)  +\mathrm{D}_{f}\left(
p_{1}\Vert q_{1}\right) \\
&  =\sum_{i=0,1}\sum_{x\in\mathcal{X}}g_{f}\left(  p(x)\mathrm{tr}%
\,\pi_{\mathcal{H}_{i}}\Gamma(\delta_{x}),q(x)\mathrm{tr}\,\pi_{\mathcal{H}%
_{i}}\Gamma(\delta_{x})\right)  =\sum_{x\in\mathcal{X}}\sum_{i=0,1}%
\mathrm{tr}\,\pi_{\mathcal{H}_{i}}\Gamma(\delta_{x})g_{f}\left(
p(x),q(x)\right) \\
&  =\sum_{x\in\mathcal{X}}g_{f}\left(  p(x),q(x)\right)  =\mathrm{D}%
_{f}(p\Vert q).
\end{align*}
Thus,
\[
\inf\mathrm{D}_{f}(p\Vert q)=\inf\left\{  \mathrm{D}_{f}\left(  p_{0}\Vert
q_{0}\right)  +\mathrm{D}_{f}\left(  p_{1}\Vert q_{1}\right)  \right\}
\geq\inf\mathrm{D}_{f}\left(  p_{0}\Vert q_{0}\right)  +\inf\mathrm{D}%
_{f}\left(  p_{1}\Vert q_{1}\right)  ,
\]
which leads to the asserted inequality. After all, we have (\ref{Dmax-sum}).
\end{proof}

\begin{lemma}
\label{lem:proper}Suppose $f$ satisfies (FC). Then the convex function
$(\rho,\sigma)\rightarrow\mathrm{D}_{f}^{\max}\left(  \rho||\sigma\right)  $
is proper. Thus, it is nowhere $-\infty$.
\end{lemma}

\begin{proof}
An improper convex function is necessarily infinite except perhaps at relative
boundary points of its effective domain (Theorem 7.2 of \cite{Rockafellar}).
But \
\[
\mathrm{D}_{f}^{\max}\left(  p\Vert p\right)  =\mathrm{D}_{f}\left(  p\Vert
p\right)  =\sum_{x\in\mathcal{X}}p(x)f(1)
\]
is finite. Thus $\mathrm{D}_{f}^{\max}$ cannot be improper.
\end{proof}

\begin{theorem}
\label{th:Df-finite}Suppose $f$ satisfies (FC). Then $\mathrm{D}_{f}^{\max
}\left(  \rho||\sigma\right)  <\infty$ only in the following four cases.

(i) $\hat{f}(0)<\infty$ and $f(0)<\infty$;

(ii)$\hat{f}(0)<\infty$, $f(0)=\infty$, and $\mathrm{supp}\,\rho
\supset\mathrm{supp}\,\sigma$;

(iii) $\hat{f}(0)=\infty$ , $f(0)<\infty$, and $\mathrm{supp}\,\rho
\subset\mathrm{supp}\,\sigma$;

(iv) $\hat{f}(0)=\infty$ , $f(0)=\infty$, and $\mathrm{supp}\,\rho
=\mathrm{supp}\,\sigma$.
\end{theorem}

\begin{proof}
In all the cases, if $\left(  \Gamma,\left\{  p,q\right\}  \right)  $ is the
minimal reverse test of $\left\{  \rho,\sigma\right\}  $, $\mathrm{D}%
_{f}\left(  p||q\right)  <\infty$. Thus below we show $\mathrm{D}_{f}^{\max
}\left(  \rho||\sigma\right)  =\infty$ in the case where these conditions are
not true.

Suppose $\hat{f}(0)=\infty$ and $\mathrm{supp}\,\rho\not \subset
\mathrm{supp}\,\sigma$. Then by (\ref{Dmax=}) $\mathrm{D}_{f}^{\max}(\rho
\Vert\sigma)=\infty$, since
\[
\mathrm{D}_{f}^{\max}(\tilde{\rho}\Vert\sigma)\geq\mathrm{D}_{ar+b}^{\max
}(\tilde{\rho}\Vert\sigma)=a\mathrm{tr}\,\tilde{\rho}+b\mathrm{tr}%
\,\sigma>-\infty,
\]
where $a$, $b$ is chosen so that $f(r)\geq ar+b$, $r\geq0$. Suppose
$f(0)=\infty$ and $\mathrm{supp}\,\rho\not \supset \mathrm{supp}\,\sigma$.
Then $\mathrm{D}_{f}^{\max}(\rho\Vert\sigma)=\infty$ is concluded by replacing
$f$ by $\hat{f}$ in the above argument.
\end{proof}

\section{When $f$ is operator convex}

\label{sec:operator-convex}

\subsection{Closed formula}

In this section, we suppose that $f$ is operator convex and $f(0)=0$ in
addition to satisfying (FC):

\begin{description}
\item[(F)] $f$ is proper, lower semicontinuous, and operator convex. In
addition, $\mathrm{dom}$\thinspace$f(x)=$ $[0,\infty)$ and $f(0)=0$.
\end{description}

If this is true and $\hat{f}(0)<\infty$, by Proposition\thinspace
\ref{prop:lowner},
\begin{equation}
f(r)=\hat{f}(0)r+f_{0}(r), \label{f=ay+h-2}%
\end{equation}
where $f_{0}(r)$ satisfies (F) and is operator monotone decreasing.

When $\mathrm{supp\,}\sigma\supset\mathrm{supp\,}\rho$, by the correspondence
(\ref{reverse}) and (\ref{reverse-2}),
\begin{equation}
\mathrm{D}_{f}^{\max}(\rho\Vert\sigma)=\inf_{\left\{  M_{x}\right\}  ,\left\{
r_{x}\right\}  _{\ast}}\left\{  \sum_{x\in\mathcal{X}}f\left(  r_{x}\right)
\mathrm{tr}\,\sigma M_{x}\,\,;\,\,\sum_{x\in\mathcal{X}}r_{x}M_{x}%
=d(\rho,\sigma),\sum_{x\in\mathcal{X}}M_{x}=\mathbf{1}\,\right\}
\label{Df-max-r}%
\end{equation}
Here we use Naimark extension. Denoting the extended space by $\mathcal{H}%
^{\prime}$, and letting $V$ be an isometry from $\mathcal{H}$ (,where $\rho$,
$\sigma$, etc. are living in) into $\mathcal{H}^{\prime}$, there is a tuple of
mutually orthogonal projectors $\left\{  E_{x}\right\}  $ in $\mathcal{H}%
^{\prime}$ with$VE_{x}V^{\dagger}=M_{x}$. Therefore,
\begin{align*}
\sum_{x\in\mathcal{X}}f\left(  r_{x}\right)  \mathrm{tr}\,\sigma M_{x}  &
=\,\mathrm{tr}\,\sigma Vf\left(  \sum_{x\in\mathcal{X}}r_{x}E_{x}\right)
V^{\dagger}\\
&  \geq\mathrm{tr}\,\sigma f\left(  \sum_{x\in\mathcal{X}}r_{x}VE_{x}%
V^{\dagger}\right)  =\mathrm{tr}\,\sigma f\left(  \sum_{x\in\mathcal{X}}%
r_{x}M_{x}\right)  =\mathrm{tr}\,\sigma f\left(  d(\rho,\sigma)\right)  ,
\end{align*}
where the inequality in the second line is by Jensen's inequality,
Proposition\thinspace\ref{prop:jensen} (Note $X\rightarrow VXV^{\dagger}$ is a
positive unital map into $\mathcal{B}\left(  \mathcal{H}\right)  $). The
identity is true if $M_{x}$'s are mutually orthogonal projectors and
$\mathcal{H}^{\prime}=\mathcal{H}$, i.e., if the reverse test is minimal.
Thus,
\[
\mathrm{D}_{f}^{\max}(\rho\Vert\sigma)=\mathrm{tr}\,\sigma f\left(
d(\rho,\sigma)\right)
\]
and the identity is achieved by the minimal test.

Next, suppose \ $\mathrm{supp\,}\sigma\not \supset \mathrm{supp\,}\rho$. If
$\hat{f}(0)=\infty$, by Theorem\thinspace\ref{th:Df-finite}, $\mathrm{D}%
_{f}^{\max}(\rho\Vert\sigma)=\infty$ . If $\hat{f}(0)<\infty$, we can apply
(\ref{Dmax=}). After all:

\begin{theorem}%
\begin{align}
\mathrm{D}_{f}^{\max}\left(  \rho||\sigma\right)   &  =\mathrm{tr}%
\,\sigma\,f\,\left(  d(\tilde{\rho},\sigma)\right)  +\hat{f}(0)\mathrm{tr}%
\,\left(  \rho-\,\tilde{\rho}\right) \label{Dmax-formula}\\
&  =\mathrm{tr}\,\sigma\,f\,\left(  \sigma^{-1}\tilde{\rho}\right)  +\hat
{f}(0)\mathrm{tr}\,\left(  \rho-\,\tilde{\rho}\right) \nonumber
\end{align}
holds if (F) is true. (If $\hat{f}(0)=\infty$, both ends are $\infty$ and the
identity holds.) The minimum is achieved by the minimal reverse test of
$\left\{  \rho,\sigma\right\}  $.
\end{theorem}

\subsection{Examples}

\label{subsec:examples}

Throughout this subsection, we suppose $\mathrm{supp}\,\sigma\supset
\mathrm{supp}\,\rho$. With $f_{\mathrm{KL}}(r):=r\log r,$
\begin{align*}
\mathrm{D}_{f_{\mathrm{KL}}}^{\max}\left(  \rho||\sigma\right)   &
=\mathrm{tr}\,\sigma\sigma^{-1}\rho\left(  \log\,\sigma^{-1}\rho\right) \\
&  =\mathrm{tr}\,\rho\log\,\sigma^{-1}\rho\\
&  =\mathrm{tr}\,\,\rho\log\,\rho^{1/2}\sigma^{-1}\rho^{1/2}.
\end{align*}
This quantity, corresponding to Kullback-Leibler divergence, had been studied
by various authors \cite{Belavkin}\cite{HammersleyBelavkin}\cite{HiaiPetz}%
\cite{Hayashi}\cite{Jencova:03}. The relation to the reverse test problem is
first pointed out by \cite{Matsumoto:05}.

Define $f_{\alpha}(r):=\left(  \pm\right)  r^{\alpha}$, where the sign is
chosen so that the function is convex on the positive half\thinspace
-\thinspace line. This is operator convex if $\alpha\in\left[  -1,1\right]
/\left\{  0\right\}  $, and \
\[
\mathrm{D}_{f_{\alpha}}^{\max}\left(  \rho||\sigma\right)  =\mathrm{tr}%
\,\sigma f_{\alpha}\left(  \sigma^{-1}\rho\right)  .
\]
One can check the following identity, which is a special case of
$\mathrm{D}_{f}^{\max}=\mathrm{D}_{\hat{f}}^{\max}$:
\[
\mathrm{D}_{f_{\alpha}}^{\max}\left(  \rho||\sigma\right)  =\mathrm{D}%
_{f_{1-\alpha}}^{\max}\left(  \sigma||\rho\right)  .
\]

\subsection{Non-commutative perspective}

\label{sec:operator-ineq}

When (F) is true, the following operator valued quantity $g_{f}(\rho,\sigma)$,
called \textit{non-commutative perspective} \cite{Ebadian}\cite{Effros},
satisfies $\mathrm{tr}\,g_{f}(\rho,\sigma)=\mathrm{D}_{f}^{\max}(\rho
\Vert\sigma)$ and some operator version of properties of $\mathrm{D}_{f}%
^{\max}$ has:
\begin{align*}
&  g_{f}(\rho,\sigma)\\
&  :=\left\{
\begin{array}
[c]{cc}%
\sqrt{\sigma}f\left(  \sigma^{-1/2}\rho\sigma^{-1/2}\right)  \sqrt{\sigma}, &
\text{if }\mathrm{supp}\,\sigma\supset\mathrm{supp}\,\rho,\\
g_{f}(\tilde{\rho},\sigma)+\hat{f}(0)(\rho-\tilde{\rho}), & \text{if
}\mathrm{supp}\,\sigma\not \supset \mathrm{supp}\,\rho\text{, }\hat
{f}(0)<\infty,\\
\text{undefined,} & \text{otherwise.}%
\end{array}
\right.
\end{align*}
In the second case, since $\hat{f}(0)<\infty$, by (\ref{f=ay+h-2}),
\begin{align}
g_{f}(\rho,\sigma)  &  =g_{f}(\tilde{\rho},\sigma)+\hat{f}(0)\left(
\rho-\tilde{\rho}\right)  =g_{f_{0}}(\tilde{\rho},\sigma)+\hat{f}%
(0)\rho\nonumber\\
&  =\inf\left\{  g_{f_{0}}\left(  \rho_{\ast},X\right)  +\hat{f}(0)\rho
;0\leq\rho_{\ast}\leq\rho,\mathrm{supp}\,X\supset\mathrm{supp}\,\rho_{\ast
}\right\}  . \label{g_f=inf}%
\end{align}

\begin{remark}
In \cite{Ebadian}\cite{Effros}, they define $g_{f}(\rho,\sigma)$ only for the
case where $\mathrm{supp}\,\sigma\supset\mathrm{supp}\,\rho$, and proves
various properties of the quantity including ones presented below.
\end{remark}

The most important one, which is used later, is operator version of (D1):

\begin{lemma}
\label{lem:gf-monotone}(i) For any positive trace preserving map $\Lambda$, we
have
\begin{equation}
\Lambda\left(  g_{f}(\rho,\sigma)\right)  \geq g_{f}\left(  \Lambda\left(
\rho\right)  ,\Lambda(\sigma)\right)  . \label{f(L)<Lf}%
\end{equation}

(ii) If $g_{f}\left(  \Lambda\left(  \rho\right)  ,\Lambda(\sigma)\right)
=\Lambda\left(  g_{f}(\rho,\sigma)\right)  $, then $\Lambda(\tilde{\rho})$ is
the largest positive operator supported on $\Lambda(\sigma)$ and majorized by
$\Lambda\left(  \rho\right)  $. Thus,%
\begin{equation}
\Lambda\left(  g_{f}(\tilde{\rho},\sigma)\right)  =g_{f}\left(  \Lambda
(\tilde{\rho}),\Lambda(\sigma)\right)  . \label{fL=Lf}%
\end{equation}

\end{lemma}

\begin{proof}
For a given positive trace preserving map $\Lambda$, define
\begin{equation}
\Lambda_{\sigma}\left(  X\right)  :=\left\{  \Lambda(\sigma)\right\}
^{-1/2}\Lambda\left(  \sigma^{1/2}X\sigma^{1/2}\right)  \left\{
\Lambda(\sigma)\right\}  ^{-1/2}, \label{def-Ls}%
\end{equation}
which is a positive unital map into $\mathcal{B}\left(  \mathrm{supp\,}%
\Lambda(\sigma)\right)  $:%
\begin{equation}
\Lambda_{\sigma}(\mathbf{1})=\pi_{\Lambda(\sigma)}, \label{L-unital}%
\end{equation}
and
\begin{equation}
\Lambda_{\sigma}\left(  d(\rho,\sigma)\right)  =d\left(  \Lambda(\rho
),\Lambda(\sigma)\right)  . \label{jencova}%
\end{equation}

\end{proof}

If $\mathrm{supp\,}\sigma\supset\mathrm{supp\,}\rho$, since $\Lambda$ is
positive, $\Lambda(\rho)$ is supported on $\mathrm{supp}$ $\Lambda(\sigma)$
and $d\left(  \Lambda(\rho),\Lambda(\sigma)\right)  $ exists. Also,
\begin{align}
g_{f}  &  \left(  \Lambda(\rho),\Lambda(\sigma)\right)  \underset{(a)}%
{=}\Lambda(\sigma)^{1/2}\,f\,\left(  \Lambda_{\sigma}\left(  \,d(\rho
,\sigma)\,\right)  \,\right)  \,\Lambda(\sigma)^{1/2}\nonumber\\
&  \underset{(b)}{\leq}\Lambda(\sigma)^{1/2}\,\Lambda_{\sigma}\left(
\,f\,\left(  \,d(\rho,\sigma)\,\right)  \,\right)  \,\Lambda(\sigma
)^{1/2}=\Lambda\left(  \sigma^{1/2}f\left(  d(\rho,\sigma)\right)
\sigma^{1/2}\right)  , \label{gf(LL)}%
\end{align}
where (a) and (b) is by (\ref{jencova}) and Proposition\thinspace
\ref{prop:jensen}, respectively. If $\mathrm{supp\,}\sigma\not \supset
\mathrm{supp\,}\rho$ and $\hat{f}(0)<\infty$,
\begin{align*}
&  g_{f}\left(  \Lambda(\rho),\Lambda(\sigma)\right) \\
&  =\inf_{\rho_{\ast}}\left\{  g_{f_{0}}\left(  \rho_{\ast},\Lambda
(\sigma)\right)  +\hat{f}(0)\Lambda(\rho);\text{ }\Lambda(\rho)\geq\rho_{\ast
}\geq0,\text{ }\mathrm{supp}\,\Lambda(\sigma)\supset\mathrm{supp}\,\rho_{\ast
}\right\} \\
&  \leq g_{f_{0}}\left(  \Lambda(\tilde{\rho}),\Lambda(\sigma)\right)
+\hat{f}(0)\Lambda\left(  \rho\right) \\
&  \leq\Lambda\left(  g_{f_{0}}\left(  \tilde{\rho},\,\sigma\right)  \right)
+\hat{f}(0)\Lambda\left(  \rho\right) \\
&  =\Lambda\left(  g_{f_{0}}\left(  \tilde{\rho},\,\sigma\right)  +\hat
{f}(0)\rho\right)  =\Lambda\left(  g_{f}(\rho,\sigma)\right)  ,
\end{align*}
where the inequality in the fourth line is by (\ref{gf(LL)}) (Recall
$\tilde{\rho}$ as of (\ref{rho-tilde}) is supported on $\mathrm{supp}\,\sigma$).

Therefore, if $g_{f}\left(  \Lambda(\rho),\Lambda(\sigma)\right)
=\Lambda\left(  g_{f}(\rho,\sigma)\right)  $, $\Lambda(\tilde{\rho})$ should
achieve the infimum in the second line. Thus the first statement of (ii) is
true. Then,
\[
g_{f}\left(  \Lambda(\rho),\Lambda(\sigma)\right)  =g_{f}\left(
\Lambda(\tilde{\rho}),\Lambda(\sigma)\right)  +\hat{f}(0)\Lambda\left(
\rho-\tilde{\rho}\right)  .
\]
Equating this to $\Lambda\left(  g_{f}(\rho,\sigma)\right)  =\Lambda\left(
g_{f}(\tilde{\rho},\sigma)\right)  +\hat{f}(0)\Lambda(\tilde{\rho})$, we have
(\ref{fL=Lf}).

Operator versions of (\ref{Dmax-homogeneous}) and (\ref{Dmax-sum}) are
trivial. Thus next we show the operator versions of (\ref{D-convex}):
\begin{equation}
g_{f}(\rho,\sigma)\leq\sum_{i}c_{i}\,g_{f}\left(  \rho_{i},\sigma_{i}\right)
, \label{D'-convex}%
\end{equation}
where $\rho:=\sum_{i}c_{i}\rho_{i}$, $\sigma:=\sum_{i}c_{i}\sigma_{i}$,
$\sum_{i}c_{i}=1$, and $c_{i}\geq0$. (\ref{D'-convex}) for the case
$\mathrm{supp}\,\sigma\supset\mathrm{supp}\,\rho$ is known \cite{Ebadian}. if
$\hat{f}(0)<\infty$,
\begin{align*}
&  \sum_{i}c_{i}\,g_{f}\left(  \rho_{i},\sigma_{i}\right)  =\sum_{i}%
c_{i}\,g_{f}\left(  \tilde{\rho}_{i},\sigma_{i}\right)  +\hat{f}(0)\sum
_{i}c_{i}\left(  \rho_{i}-\tilde{\rho}_{i}\right)  \,\\
&  \geq g_{f_{0}}\left(  \sum_{i}c_{i}\,\tilde{\rho}_{i},\sum_{i}c_{i}%
\sigma_{i}\right)  +\hat{f}(0)\sum_{i}c_{i}\rho_{i}%
\end{align*}
Above, since $\mathrm{supp}\,\sigma=\mathrm{span}\{\mathrm{\,}\bigcup
\mathrm{supp}\,\sigma_{i}\}$, $\sum_{i}c_{i}\,\tilde{\rho}_{i}$ is supported
on $\mathrm{supp}\,\sigma$. Thus the last end is well-defined. Also, $\sum
_{i}c_{i}\tilde{\rho}_{i}\leq\sum_{i}c_{i}\rho_{i}=\rho$. Thus, The last end
is bounded from below by%

\begin{align*}
&  \inf\left\{  g_{f_{0}}\left(  \rho_{\ast},\sigma\right)  +\hat{f}%
(0)\rho;\rho\geq\rho_{\ast}\geq0,\mathrm{supp}\,\sigma\supset\mathrm{supp}%
\,\rho_{\ast}\right\} \\
&  =g_{f_{0}}(\tilde{\rho},\sigma)+\hat{f}(0)\rho=g_{f}(\rho,\sigma),
\end{align*}
concluding (\ref{D'-convex}).

Lastly, the analogue of (\ref{Dmax-monotone}) is
\begin{equation}
g_{f}(\rho,\sigma)\leq g_{f}\left(  \rho,X\right)  ,\,\,X\geq\sigma
.\label{D'-monotone}%
\end{equation}
This is proved as follows. Since $X\geq\sigma$, $C:=\sigma^{1/2}X^{-1/2}$
satisfies
\[
CX^{1/2}=\sigma^{1/2},\,\,\left\Vert C\right\Vert \leq1.
\]
If $\mathrm{supp}\,X\supset\mathrm{supp}\,\sigma\supset\mathrm{supp}\,\rho$,
\begin{align*}
g_{f}(\rho,\sigma) &  =X^{1/2}C^{\dagger}\,f\left(  d(\rho,\sigma)\right)
X^{1/2}C\\
&  \geq X^{1/2}\,\,f\,\left(  C^{\dagger}d(\rho,\sigma)C\right)
\,X^{1/2}=\,X^{1/2}\,\,f\,\left(  X^{-1/2}\rho\,X^{-1/2}\right)  X^{1/2}\\
&  =g_{f}\left(  \rho,X\right)  ,
\end{align*}
where the inequality in the second line\ is due to Proposition\thinspace
\ref{prop:convex-cfc}. If $\hat{f}(0)<\infty$ and $\mathrm{supp}%
\,\sigma\not \supset \mathrm{supp}\,\rho$,
\begin{align*}
g_{f}(\rho,\sigma) &  =g_{f_{0}}(\tilde{\rho},\sigma)+\hat{f}(0)\rho\geq
g_{f_{0}}\left(  \tilde{\rho},X\right)  +\hat{f}(0)\rho\\
&  \geq\inf\left\{  g_{f_{0}}\left(  \rho_{\ast},X\right)  +\hat{f}%
(0)\rho;0\leq\rho_{\ast}\leq\rho,\mathrm{supp}\,X\supset\mathrm{supp}%
\,\rho_{\ast}\right\}  \\
&  =g_{f}\left(  \rho,X\right)  .
\end{align*}

\begin{remark}
Here, $\tilde{\rho}$ is not necessarily the largest element of the set
\[
\{\rho_{\ast};0\leq\rho_{\ast}\leq\rho,\mathrm{supp}\,X\supset\mathrm{supp}%
\,\rho_{\ast}\}.
\]
Thus, in general, the equality between the second and the third line does not hold.
\end{remark}

\subsection{Relation to RLD Fisher metric}

$\mathrm{D}_{f}^{\max}$ is closely related to RLD Fisher metric\thinspace\
\[
J_{\rho}^{R}\left(  X,Y\right)  :=\mathrm{tr}\,X\rho^{-1}Y,
\]
where $X$ and $Y$ are self -\thinspace adjoint operators living in the support
of $\rho>0$, $\ $with $\mathrm{tr}\,X=\mathrm{tr}\,Y=0$. This quantity plays
important role in quantum statistical estimation theory \cite{Holevo}, and is
the largest monotone metric on the space of density operators \cite{Petz}.
Also, this quantity is the solution to infinitesimal version of reverse test
\cite{Matsumoto:05}; The triple $\left(  \Gamma,p,v\right)  $ of the positive
trace preserving map $\Gamma$ from positive measures to self\thinspace
-\thinspace adjoint operators, the probability distribution over the finite
set $\mathcal{X}$, and the real valued function over $\mathcal{X}$ with
$\sum_{x\in\mathcal{X}}v(x)$, is said to be \textit{reverse estimation of
}$\left\{  \rho,X\right\}  $ iff
\[
\Gamma\left(  p\right)  =\rho,\Gamma\left(  v\right)  =X.
\]
Then
\[
J_{\rho}^{R}\left(  X,X\right)  =\inf\left\{  \sum_{x\in\mathcal{X}}%
\frac{\left\{  v(x)\right\}  ^{2}}{p(x)};\left(  \Gamma,p,v\right)  \text{ is
a reverse estimation of }\left\{  \rho,X\right\}  \right\}  .
\]
Here, the function minimized is called \textit{Fisher information} and plays
significant roll in point estimation.

This problem reduces to our reverse test problem of $\left\{  \rho+\varepsilon
X,\rho\right\}  $, where $\varepsilon$ is chosen so that $\rho+\varepsilon
X\geq0$, and $f(r)=$ $\left(  1-r\right)  ^{2}$. Since $f(r)-1$ satisfies (F),
\ (\ref{Dmax-formula}) shows the above identity, and the optimal reverse
estimation is the one such that $\left(  \Gamma,\,\left\{  p+\varepsilon
v,p\right\}  \right)  $ is the minimal reverse test of $\left\{
\rho+\varepsilon X,\rho\right\}  $.

Below, we prove
\begin{equation}
f^{\prime\prime}(1)J_{\rho}^{R}\left(  X,X\right)  ,=\left.  \frac
{\mathrm{d}^{2}}{\mathrm{d}\varepsilon\,^{2}}\mathrm{D}_{f}^{\max}\left(
\rho+\varepsilon X\Vert\rho\right)  \right\vert _{\varepsilon=0}=\left.
\frac{\mathrm{d}^{2}}{\mathrm{d}\varepsilon\,^{2}}\mathrm{D}_{f}^{\max}\left(
\rho\Vert\rho+\varepsilon X\right)  \right\vert _{\varepsilon=0}
\label{Dmax=RLD}%
\end{equation}
when $f^{\prime\prime\prime}$ exists and uniformly bounded in the sense that
\[
\left\vert f^{\prime\prime\prime}(x)\right\vert <c,\,\,\exists\varepsilon
>0\forall x\in\left(  1-\varepsilon,1+\varepsilon\right)  .\,
\]
(Differentiating (\ref{Dmax-formula}) twice, one can also obtain
(\ref{Dmax=RLD}).

The key observation is: The minimal reverse test of $\left\{  \rho+\varepsilon
X,\rho\right\}  $ are the same for all $\varepsilon>0$ with $\rho+\varepsilon
X\geq0,$ Therefore, differentiating the both sides of
\[
\mathrm{D}_{f}^{\max}\left(  \rho+\varepsilon X||\rho\right)  =\mathrm{D}%
_{f}\left(  p+\varepsilon v||p\right)
\]
twice, well-known relation between Fisher information and $f$-divergence leads
to the first identity. The second identity follows from Corollary\thinspace
\ref{cor:minimal-same} (which will be shown later) that states the minimal
reverse test of $\left\{  \rho+\varepsilon X,\rho\right\}  $ and $\left\{
\rho,\rho+\varepsilon X\right\}  $ are identical.

\subsection{Essential uniqueness of optimal reverse test}

\label{sec:unique}

In this section, we show that any optimal reverse test is essentially
identical to the minimal reverse test, provided that (F) is satisfied. First,
we show some technical lemmas. (They themselves are of interest. We show
another application of them in the next section)

A function $f$ with (F), by Theorem\thinspace8.1 of \cite{HiaiMosonyiPetzBeny}%
, is written as
\begin{equation}
f(r)=cr+br^{2}+\int_{\left(  0,\infty\right)  }\left(  \frac{r}{1+\lambda
}+\psi_{\lambda}(r)\right)  \mathrm{d}\mu\left(  \lambda\right)  ,
\label{lowner-2}%
\end{equation}
where $c$ is a real number, $b>0$, $\mu$ is a positive Borel measure with
$\int_{\left(  0,\infty\right)  }\frac{\mathrm{d}\mu\left(  \lambda\right)
}{\left(  1+\lambda\right)  ^{2}}<\infty$, and $\psi_{\lambda}(r)\colon
=-\frac{r}{\lambda+r}$. \ 

In what follows, we suppose
\begin{equation}
\lbrack0,\infty)\subset\mathrm{supp\,}\mu\cup\left\{  0\right\}  ,
\label{spec<support}%
\end{equation}
where $\mathrm{supp\,}\mu$ is the set of all points $r$ having property that
$\mu\left(  U\right)  >0$ for any open set $U$ containing $\lambda$ (see
Theorem 2.2.1 and Definition 2.2.1 of \cite{Parthasarathy}.) $r\ln r$,
$\left(  \pm1\right)  r^{\alpha}\,$\ ($-1\leq\alpha\leq1$, $\alpha\neq1$)
satisfies (\ref{spec<support}) (see Example 8.3, \cite{HiaiMosonyiPetzBeny}).

\begin{lemma}
\label{lem:Lh=hL}Suppose $f$ satisfies (F) and (\ref{spec<support}). Suppose
also $\mathrm{D}_{f}^{\max}\left(  \rho||\sigma\right)  <\infty$. Let
$\Lambda$ be a positive trace preserving map. Then,
\begin{equation}
\mathrm{D}_{f}^{\max}\left(  \rho||\sigma\right)  =\mathrm{D}_{f}^{\max
}\left(  \Lambda\left(  \rho\right)  ||\Lambda(\sigma)\right)  \label{D=DL}%
\end{equation}
implies
\begin{equation}
\Lambda_{\sigma}\left(  h\left(  d(\tilde{\rho},\sigma)\right)  \right)
=h\left(  \Lambda_{\sigma}\left(  d(\tilde{\rho},\sigma)\right)  \right)  .
\label{Lh=hL}%
\end{equation}
Here, $\Lambda_{\sigma}$ is a subunital positive map defined by (\ref{def-Ls}%
), and $h$ is an arbitrary function on $[0,\infty)$.

Conversely, if (\ref{Lh=hL}) holds, (\ref{D=DL}) holds for any $f$ with (F).
In fact, for any function $h$ on $[0,\infty)$, $\ $%
\begin{equation}
\mathrm{tr}\,\sigma h(d(\tilde{\rho},\sigma))=\mathrm{tr}\,\Lambda
(\sigma)h(d\left(  \Lambda(\tilde{\rho}),\Lambda(\sigma)\right)  ).
\label{D=DL-2}%
\end{equation}

\end{lemma}

\begin{proof}
By (\ref{f(L)<Lf}), (\ref{D=DL}) and $\mathrm{D}_{f}^{\max}\left(
\rho||\sigma\right)  <\infty$ implies
\begin{equation}
\Lambda\left(  g_{f}(\rho,\sigma)\right)  =g_{f}\left(  \Lambda(\rho
),\Lambda(\sigma)\right)  . \label{f(L)=Lf}%
\end{equation}

First, suppose $\mathrm{supp}\,\rho\subset\mathrm{supp}\,\sigma$. Observe, by
(\ref{lowner-2}),
\begin{equation}
g_{f}(\rho,\sigma)=c\rho+b\,g_{f_{2}}(\rho,\sigma)+\int_{\left(
0,\infty\right)  }\left(  \frac{\rho}{1+\lambda}+g_{\psi_{\lambda}}%
(\rho,\sigma)\right)  \mathrm{d}\mu\left(  \lambda\right)  \label{lowner-3}%
\end{equation}
where $f_{2}(r):=r^{2}$. Since $f_{2}$ and $\psi_{\lambda}(r)=-\frac
{r}{\lambda+r}$ satisfies (F), (\ref{f(L)<Lf}). (\ref{f(L)=Lf}) and
(\ref{spec<support}) lead to
\[
\Lambda\left(  g_{\psi_{t}}(\rho,\sigma)\right)  =g_{\psi_{t}}\left(
\Lambda\left(  \rho\right)  ,\Lambda(\sigma)\right)  \,,\,\,\forall t>0.
\]
This, by Proposition\thinspace\ref{prop:f-finite}, implies%
\begin{equation}
\Lambda\left(  g_{h}(\rho,\sigma)\right)  =g_{h}\left(  \Lambda\left(
\rho\right)  ,\Lambda(\sigma)\right)  , \label{L(h(d))=}%
\end{equation}
which, using (\ref{jencova}), implies (\ref{Lh=hL}).

Next, suppose $\mathrm{supp}\,\rho\not \subset \mathrm{supp}\,\sigma$. Then
for $\mathrm{D}_{f}^{\max}\left(  \rho||\sigma\right)  <\infty$ to be true,
$\hat{f}(0)<\infty$ should hold. Therefore, by (\ref{fL=Lf}),%
\[
\,\Lambda\left(  g_{f}(\tilde{\rho},\sigma)\right)  =g_{f}\left(
\Lambda(\tilde{\rho}),\Lambda(\sigma)\right)  .
\]
Therefore, using the parallel argument as above, we have (\ref{Lh=hL}).

The second assertion of the theorem is proved by straightforward computation.
\end{proof}

\begin{lemma}
\label{th:preserve}Suppose $f$ satisfies (F) and (\ref{spec<support}). Suppose
also $\mathrm{D}_{f}^{\max}\left(  \rho||\sigma\right)  <\infty$. Let
$\Lambda$ be a positive trace preserving map. Also, let $\left(
\Gamma,\left\{  p,q\right\}  \right)  $ and $\left(  \Gamma^{\prime},\left\{
p^{\prime},q^{\prime}\right\}  \right)  $ be the minimal reverse test of
$\left\{  \rho,\sigma\right\}  $ and $\left\{  \Lambda\left(  \rho\right)
,\Lambda(\sigma)\right\}  $, respectively. Then, (\ref{D=DL}) holds iff
\begin{equation}
\Lambda\left(  \Gamma(\delta_{x})\right)  =\Gamma^{\prime}(\delta
_{x}),\,\,\left\{  p,q\right\}  =\left\{  p^{\prime},q^{\prime}\right\}  .
\label{LG=G'}%
\end{equation}

\end{lemma}

\begin{proof}
Since `if' is trivial, we prove `only if'. Recall the minimal reverse test is
given by
\begin{align*}
\Gamma(\delta_{x})  &  =\frac{1}{\mathrm{tr}\,\left(  \rho-\tilde{\rho
}\right)  }\left(  \rho-\tilde{\rho}\right)  ,\\
\Gamma(\delta_{x})  &  =\frac{1}{\mathrm{tr}\,\sigma P_{x}}\sqrt{\sigma}%
P_{x}\sqrt{\sigma},\,\left(  x\neq x_{0}\right)  ,
\end{align*}
where $d(\tilde{\rho},\sigma)=\sum_{x}d_{x}P_{x}$, where $d_{x}$ and $P_{x}$
is an eigenvalue and projection onto eigenspace, respectively.

Let $P_{x}^{\prime}:=\Lambda_{\sigma}\left(  P_{x}\right)  $, and applying
(\ref{Lh=hL}) with
\[
h_{0}(r):=\left\{
\begin{array}
[c]{cc}%
1, & \left(  r=d_{x}\right)  ,\\
0, & \text{otherwise}.
\end{array}
\right.
\]
we have
\[
P_{x}^{\prime}=\Lambda_{\sigma}\left(  P_{x}\right)  =\Lambda_{\sigma}\left(
h_{0}\left(  d(\tilde{\rho},\sigma)\right)  \right)  =h_{0}\left(
\Lambda_{\sigma}\left(  d(\tilde{\rho},\sigma)\right)  \right)  .
\]
Since eigenvalues of $h_{0}\left(  \Lambda_{\sigma}\left(  d\right)  \right)
$ are either 0 or 1, $P_{x}^{\prime}$ is a projector. Since
\[
d\left(  \Lambda(\tilde{\rho}),\Lambda(\sigma)\right)  =\Lambda_{\sigma
}\left(  d(\tilde{\rho},\sigma)\right)  =\sum_{x}d_{x}\Lambda_{\sigma}\left(
P_{x}\right)  =\sum_{x}d_{x}P_{x}^{\prime},
\]
$P_{x}^{\prime}$s are the projectors onto the eigenspaces of $d^{\prime}$, and
$\mathrm{spec\,}d=\mathrm{spec\,}d^{\prime}$. Therefore, if $x\neq x_{0}$,%
\begin{align*}
\Lambda\left(  \Gamma(\delta_{x})\right)   &  =\Lambda\left(  \sqrt{\sigma
}P_{x}\sqrt{\sigma}\right)  =\sqrt{\Lambda(\sigma)}\Lambda_{\sigma}\left(
P_{x}\right)  \sqrt{\Lambda(\sigma)}\\
&  =\sqrt{\Lambda(\sigma)}P_{x}^{\prime}\sqrt{\Lambda(\sigma)}=\Gamma^{\prime
}(\delta_{x}).
\end{align*}

Since $\Lambda$ is trace preserving, this relation is easily checked for
$x=x_{0}$,
\begin{align*}
\Gamma^{\prime}(\delta_{x_{0}})  &  =\frac{1}{\mathrm{tr}\,\Lambda\left(
\rho-\tilde{\rho}\right)  }\Lambda\left(  \rho-\tilde{\rho}\right) \\
&  =\Lambda\left(  \frac{1}{\mathrm{tr}\,\left(  \rho-\tilde{\rho}\right)
}\left(  \rho-\tilde{\rho}\right)  \right)  =\Lambda\left(  \Gamma\left(
\delta_{x_{0}}\right)  \right)  .
\end{align*}
Therefore, if $\left(  \Gamma^{\prime},\left\{  p^{\prime},q^{\prime}\right\}
\right)  $ is the minimal reverse test of $\left\{  \Lambda(\rho
),\Lambda(\sigma)\right\}  $, $\ \Gamma^{\prime}=\Lambda\circ\Gamma$. \ 

Having specified the map $\Gamma^{\prime}$, next we specify $\left\{
p^{\prime},q^{\prime}\right\}  $. For $\left(  \Lambda\circ\Gamma,\left\{
p^{\prime},q^{\prime}\right\}  \right)  $ to be a reverse test of $\left\{
\Lambda(\rho),\Lambda(\sigma)\right\}  $,%
\[
\sum_{x\neq x_{0}}q^{\prime}(x)\Lambda\circ\Gamma(\delta_{x})=\Lambda
(\sigma)=\sum_{x\neq x_{0}}q(x)\Lambda\circ\Gamma(\delta_{x}).
\]
Thus we have to have, for all $x\neq x_{0}$,
\[
\sum_{x}q^{\prime}(x)\sqrt{\Lambda(\sigma)}P_{x}^{\prime}\sqrt{\Lambda
(\sigma)}=\sum_{x}q(x)\sqrt{\Lambda(\sigma)}P_{x}^{\prime}\sqrt{\Lambda
(\sigma)}.
\]
Since $P_{x}^{\prime}$ is supported on $\mathrm{supp}\,d^{\prime
}=\mathrm{supp}\,\Lambda(\sigma)$, this is equivalent to
\[
\sum_{x}q^{\prime}(x)P_{x}^{\prime}=\sum_{x}q(x)P_{x}^{\prime}.
\]
Since $P_{x}^{\prime}$s are orthogonal projectors, we have $q^{\prime}=q$.

In the same way, we can prove that $p^{\prime}(x)=p(x)$, for all $x\neq x_{0}%
$. Then by trace preserving nature of $\Lambda$, obviously $p^{\prime}\left(
x_{0}\right)  =p\left(  x_{0}\right)  $, concluding $p^{\prime}=p$. Thus we
have the assertion.
\end{proof}

In the following, we extend the notion of the minimal reverse test to the pair
$\left\{  p,q\right\}  $ of positive measures, by identifying $p$ with the
diagonal matrix, $\sum_{x}p(x)\left\vert e_{x}\right\rangle \left\langle
e_{x}\right\vert $, where $\left\{  \left\vert e_{x}\right\rangle \right\}  $
is a CONS. Let $\left(  \Upsilon_{0},\left\{  p_{0},q_{0}\right\}  \right)  $
be the minimal reverse test of $\left\{  p,q\right\}  $, where $\left\{
p_{0},q_{0}\right\}  $ are positive measures over the finite set $\mathcal{Y}%
$. Then $\Upsilon_{0}$ is in fact stochastic map, but at the same time viewed
as the positive trace preserving map sending diagonal density matrices to
diagonal density matrices.

With this correspondence, the equivalence $\tilde{p}$ of $\tilde{\rho}$ (see
(\ref{rho-tilde}) ) is in fact restriction of $p$ to $\mathrm{supp}q$, and
\begin{align*}
d\left(  \tilde{p},q\right)   &  =\sum_{x\in\mathrm{supp}q}\frac{p(x)}%
{q(x)}\left\vert e_{x}\right\rangle \left\langle e_{x}\right\vert =\sum
_{y\in\mathcal{Y}\backslash\left\{  y_{0}\right\}  }r_{y}P_{y},\\
\mathrm{supp}\,P_{y}  &  =\mathrm{span}\,\{\left\vert e_{x}\right\rangle
;p(x)/q(x)=r_{y}\}.
\end{align*}

Viewing $\Upsilon_{1}$ as a stochastic map, $y$ ($\neq y_{0}$) is mapped to
$x$ iff $p(x)/q(x)=r_{y}$, and $y_{0}$ is mapped to $x\in\left(
\mathrm{supp}q\right)  ^{c}$. (Detailed form of the transition probability is
not relevant now.)

\begin{lemma}
\label{lem:c-invert}Let $\left(  \Upsilon_{0},\left\{  p_{0},q_{0}\right\}
\right)  $ be the minimal reverse test of $\left\{  p,q\right\}  $, where
$\left\{  p_{0},q_{0}\right\}  $ are positive measures over the finite set
$\mathcal{Y}$. Then, there is a positive trace preserving map $\Upsilon
_{0}^{-}$ that invert $\Upsilon_{0}$,
\[
\Upsilon_{0}^{-}\left(  p\right)  =p_{0},\Upsilon_{0}^{-}\left(  q\right)
=q_{0}.
\]
In addition, $\Upsilon_{0}^{-}$ is deterministic. Therefore,
\[
\Upsilon_{0}^{-}\circ \Upsilon_{0}\left(  \delta_{y}\right)  =\delta_{y}.
\]

\end{lemma}

\begin{proof}
$\Upsilon_{0}^{-}$ corresponds to the following deterministic map from
$\mathcal{X}$ to $\mathcal{Y}$ : $x\in\mathrm{supp}$ $q$ is mapped to $y$ iff
$p(x)/q(x)=r_{y}$, and $\ x\in\left(  \mathrm{supp}q\right)  ^{c}$ is mapped
to $y_{0}$. .
\end{proof}

\begin{remark}
In the statistician's term, $r_{y}$ is likelihood ratio, and thus $y$ is
minimal sufficient statistic of the family $\left\{  p,q\right\}  $
\cite{Strasser}. That roughly means $y$ contains all the information about the
family $\left\{  p,q\right\}  $, and the smallest one among those having the
same property. Thus, $\left\{  p_{0},q_{0}\right\}  $ is a kind of
"compression" of $\left\{  p,q\right\}  $. In fact, the map from $\left\{
p,q\right\}  $ to $\left\{  p_{0},q_{0}\right\}  $ is deterministic, while its
inverse is noisy.
\end{remark}

Lemmas \ref{th:preserve} \ and \ref{lem:c-invert} indicate that the optimal
reverse test is essentially unique.

\begin{theorem}
\label{th:unique}Suppose $\mathrm{D}_{f}^{\max}\left(  \rho||\sigma\right)
<\infty$, where $f$ is a function with (F), and $\mu$ defined by
(\ref{lowner-2}) satisfies(\ref{spec<support}). Let $\left(  \Gamma,\left\{
p,q\right\}  \right)  $ be an optimal reverse test,
\begin{equation}
\mathrm{D}_{f}^{\max}\left(  \rho||\sigma\right)  =\mathrm{D}_{f}\left(
p||q\right)  . \label{Dmax-achieve}%
\end{equation}
If $\left(  \Gamma_{1},\left\{  p_{1},q_{1}\right\}  \right)  $ is the minimal
reverse test of $\left\{  \rho,\sigma\right\}  $, there is a CPTP map
$\Upsilon_{0}$ and $\Upsilon_{0}^{-}$ with%
\begin{equation}
\left\{  p,q\right\}  \overset{\Upsilon_{0}}{\underset{\Upsilon_{0}^{-}%
}{\leftrightarrows}}\left\{  p_{1},q_{1}\right\}  , \label{<->}%
\end{equation}
and $\Upsilon_{0}^{-}$ is deterministic: $\Upsilon_{0}^{-}\circ \Upsilon
_{0}\left(  \delta_{y}\right)  =\delta_{y}$.Therefore,
\begin{equation}
\Gamma_{1}=\Gamma\circ \Upsilon_{0}. \label{<->2}%
\end{equation}

\end{theorem}

Before proving this, let us see its implication. (\ref{<->2}) intuitively
means that any optimal reverse test $\Gamma_{1}$ differs from the minimal one
only in its classical preprocessing ("essential uniqueness"). 

\begin{proof}
Let $\left(  \Upsilon_{0},\left\{  p_{0},q_{0}\right\}  \right)  $ be the
minimal reverse test of $\left\{  p,q\right\}  $. Then, taking recourse to
Lemma\thinspace\ref{th:preserve}, we have $\left\{  p_{1},q_{1}\right\}
=\left\{  p_{0},q_{0}\right\}  $.

Therefore,
\[
p=\Upsilon_{0}\left(  p_{1}\right)  ,\,q=\Upsilon_{0}\left(  q_{1}\right)  .
\]
Also, by Lemma\thinspace\ref{lem:c-invert}, there is a positive trace
preserving map $\Upsilon_{0}^{-}$ with
\begin{equation}
p_{1}=\Upsilon_{0}^{-}\left(  p\right)  ,q_{1}=\Upsilon_{0}^{-}\left(
q\right)  , \label{invert}%
\end{equation}
and $\Upsilon_{0}^{-}\circ \Upsilon_{0}\left(  \delta_{y}\right)  =\delta_{y}$.
\end{proof}

The following simple statement is not easy to prove directly, but quite easy
if the theorem is given.

\begin{corollary}
\label{cor:minimal-same}The minimal reverse test of $\left\{  \rho
,\sigma\right\}  $ and $\left\{  \sigma,\rho\right\}  $ are identical.
\end{corollary}

\begin{proof}
Let us consider an operator convex function $f(r)=-\sqrt{r}$ , that satisfies
(F) and (\ref{spec<support}) (Example\thinspace8.3, \cite{HiaiMosonyiPetzBeny}%
). Since $\hat{f}$ as of (\ref{f-hat}) satisfies $\hat{f}(r)=-\sqrt{r}=f(r)$,
by ( \ref{f-hat-f}), we have
\[
\mathrm{D}_{f}^{\max}\left(  \sigma\Vert\rho\right)  =\mathrm{D}_{\hat{f}%
}^{\max}(\rho\Vert\sigma)=\mathrm{D}_{f}^{\max}(\rho\Vert\sigma).
\]
Thus the minimal reverse test $\left(  \Gamma,\left\{  p,q\right\}  \right)  $
of $\left\{  \sigma,\rho\right\}  $ also achieves $\mathrm{D}_{f}^{\max}%
(\rho\Vert\sigma)$. Therefore, by Theorem\thinspace\ref{th:unique} $\left\{
p_{1},q_{1}\right\}  =\left\{  \Upsilon^{-}\left(  p\right)  ,\Upsilon
^{-}\left(  q\right)  \right\}  $, where $\left(  \Gamma_{1},\left\{
p_{1},q_{1}\right\}  \right)  $ is the minimal reverse test of $\left\{
\rho,\sigma\right\}  $, and $\Upsilon^{-}$ is a deterministic map. Exchanging
the role of $\left\{  \rho,\sigma\right\}  $ and $\left\{  \sigma
,\rho\right\}  $, there is a deterministic map $\Upsilon_{2}^{-}$ with
$\left\{  p,q\right\}  =\left\{  \Upsilon_{1}^{-}\left(  p_{1}\right)
,\Upsilon_{1}^{-}\left(  q_{1}\right)  \right\}  $. Therefore,
\[
\Gamma=\Gamma_{1}\circ \Upsilon^{-},\,\,\Gamma_{1}=\Gamma\circ \Upsilon_{1}%
^{-}.
\]
Thus we have the assertion.
\end{proof}

\subsection{Invertible reverse tests}

\label{subsec:invertible}

\begin{corollary}
\label{cor:invertible}Suppose $\mathrm{D}_{f}^{\max}\left(  \rho
||\sigma\right)  <\infty$, where $f$ is a function with (F), and $\mu$ defined
by (\ref{lowner-2}) satisfies (\ref{spec<support}). Then if there is a
measurement $M$ taking values on the finite set $\mathcal{Z}$ such that
\[
\mathrm{D}_{f}^{\max}\left(  \rho||\sigma\right)  =\mathrm{D}_{f}\left(
P_{\rho}^{M}\,||\,P_{\sigma}^{M}\right)  ,
\]
$\rho$ and $\sigma$ commute.
\end{corollary}

Intuitively, this result seems trivial: It is not possible to retrieve
classical information imbedded in quantum states perfectly, unless the quantum
states are in fact commutative (classical). However, as later turns out, this
result is generally not true if $f$ is not operator convex. A counter example
is $f(r)=\left\vert 1-r\right\vert $, and this corresponds to the total
variation distance, one of the most commonly used distance measure.

\begin{proof}
Suppose $\left\{  P_{\rho}^{M}\,,\,P_{\sigma}^{M}\right\}  $ is a positive
measures over the finite set $\mathcal{Z}$. Let $\left(  \Gamma,\left\{
p,q\right\}  \right)  $ and $\ \left(  \Upsilon_{0},\left\{  p_{0}%
,q_{0}\right\}  \right)  $ be the minimal reverse test of $\left\{
\rho,\sigma\right\}  $ and $\left\{  P_{\rho}^{M}\,,\,P_{\sigma}^{M}\right\}
$, respectively. Apply Lemma\thinspace\ref{th:preserve} considering the
measurement $M$ as a positive linear map. Then we have $\left\{  p_{0}%
,q_{0}\right\}  =\left\{  p,q\right\}  $ and
\begin{equation}
P_{\Gamma(\delta_{x})}^{M}=\Upsilon_{0}(\delta_{x}) \label{PM=Gd}%
\end{equation}

Since $\left\{  P_{\rho}^{M}\,,\,P_{\sigma}^{M}\right\}  $ are probability
distributions, by Lemma\thinspace\ref{lem:c-invert}, there is a positive trace
preserving map $\Upsilon_{0}^{-}$ with $\Upsilon_{0}^{-}\left(  \Upsilon
_{0}(\delta_{x})\right)  =\delta_{x}$. Composing them, we obtain
\[
\Upsilon_{0}^{-}\left(  P_{\Gamma(\delta_{x})}^{M}\right)  =\delta_{x}.
\]
The composition of the measurement $M$ followed by the data processing
$\Gamma_{0}^{-}$ can be viewed as a measurement, to which POVM $\{\tilde
{M}_{x}\}$ corresponds. Then, this can be rewritten as
\[
\mathrm{tr}\,\tilde{M}_{x}\Gamma(\delta_{x})=1,\,x\in\mathcal{X}\text{.}%
\]
Since $\mathrm{tr}\,\Gamma(\delta_{x})=1$, this means that $\Gamma(\delta
_{x})\Gamma\left(  \delta_{x^{\prime}}\right)  =0$ ($x^{\prime}\neq x$).
Therefore, $\rho=\sum_{x\in\mathcal{X}}p(x)\Gamma(\delta_{x})$ and
$\sigma=\sum_{x\in\mathcal{X}}q(x)\Gamma(\delta_{x})$ commute.
\end{proof}

\begin{proposition}
\label{prop:invertible-2}If $f$ satisfies (F) and $\left\vert f^{\prime
\prime\prime}(x)\right\vert <c,\,\,\exists\varepsilon>0\forall x\in\left(
1-\varepsilon,1+\varepsilon\right)  $. Let $\rho_{\varepsilon}:=\rho
+\varepsilon X>0$. Then if that there is a measurement $M_{\varepsilon}$ for
each $\varepsilon$ taking values on the finite set $\mathcal{Z}$ such that
\[
\mathrm{D}_{f}^{\max}\left(  \rho\Vert\rho_{\varepsilon}\right)
=\mathrm{D}_{f}(P_{\rho_{\varepsilon}}^{M_{\varepsilon}}\Vert\,P_{\rho
_{\varepsilon}}^{M_{\varepsilon}}),
\]
then $\rho$ and $X$ commute.
\end{proposition}

\begin{proof}
By (\ref{Dmax=RLD}) and by Section\thinspace9 of \cite{Matsumoto:14}, this is
equivalent to the existence of the measurement $M$ with
\[
J_{\rho}^{R}\left(  X,X\right)  =J_{p^{M}}\left(  v^{M},v^{M}\right)  ,
\]
where the RHS is the Fisher information of $p^{M}:=P_{\rho}^{M}$, $v^{M}%
=\frac{1}{\varepsilon}\left(  P_{\rho+\varepsilon X}^{M}-P_{\rho}^{M}\right)
$. But this is impossible unless $\rho$ and $X$ commute (see, for example,
\cite{Matsumoto:dr}).
\end{proof}

\subsection{Relation to comparison of state families}

Let $\left\{  \left\vert \hat{\varphi}_{x}\right\rangle ;x\in\mathcal{X}%
\right\}  $ be family of linearly independent state vectors. Also let
$\left\{  \tau_{x};x\in\mathcal{X}\right\}  $ be a family of density
operators. The necessary and sufficient condition for the existence of CPTP
map $\Lambda$ with \
\begin{equation}
\Lambda\left(  \tau_{x}\right)  =\left\vert \hat{\varphi}_{x}\right\rangle
\left\langle \hat{\varphi}_{x}\right\vert ,\,\forall x\in\mathcal{X}
\label{phi->phi}%
\end{equation}
have been studied by several authors. Especially, if $\tau_{x}=\left\vert
\varphi_{x}\right\rangle \left\langle \varphi_{x}\right\vert $, it is
expressed in the following very simple form.%
\[
\exists A\geq0\,\,\,\,\left\langle \varphi_{x}\right.  \left\vert
\varphi_{x^{\prime}}\right\rangle =A_{x,x^{\prime}}\left\langle \hat{\varphi
}_{x}\right.  \left\vert \hat{\varphi}_{x^{\prime}}\right\rangle
\]
(see \cite{Chefles}\cite{Uhlmann}).

Here we show this is equivalent to
\begin{equation}
\Lambda(\rho)=\hat{\rho},\Lambda(\sigma)=\hat{\sigma}, \label{r-s->r-s}%
\end{equation}
where
\begin{align*}
\rho &  :=\sum_{x}p(x)\tau_{x},\,\sigma:=\sum_{x}q(x)\tau_{x},\\
\hat{\rho}  &  :=\sum_{x}p(x)\left\vert \hat{\varphi}_{x}\right\rangle
\left\langle \hat{\varphi}_{x}\right\vert ,\,\hat{\sigma}:=\sum_{x}%
q(x)\left\vert \hat{\varphi}_{x}\right\rangle \left\langle \hat{\varphi}%
_{x}\right\vert ,
\end{align*}
and $\left\{  p,q\right\}  $ is a probability distributions over $\mathcal{X}$
such that, with $r_{x}:=p(x)/q(x)$,
\begin{equation}
r_{x}<\infty,\,\,r_{x}\neq r_{x^{\prime}},(x\neq x^{\prime}). \label{r-r}%
\end{equation}

That (\ref{phi->phi}) implies (\ref{r-s->r-s}) is trivial. To show the
opposite implication, we take recourse to Lemma\thinspace\ref{th:preserve}.

First, observe $\left(  \Gamma,\left\{  p,q\right\}  \right)  $ and $\left(
\hat{\Gamma},\left\{  p,q\right\}  \right)  $, where $\Gamma(\delta_{x}%
):=\tau_{x}$ and $\hat{\Gamma}(\delta_{x}):=\left\vert \hat{\varphi}%
_{x}\right\rangle \left\langle \hat{\varphi}_{x}\right\vert $, is a reverse
test of $\left\{  \rho,\sigma\right\}  $ and $\left\{  \hat{\rho},\hat{\sigma
}\right\}  $, respectively. In addition, the latter one is minimal, since we
had supposed that $\left\vert \hat{\varphi}_{x}\right\rangle $'s are linearly
independent and that $\left\{  p,q\right\}  $ satisfies (\ref{r-r}).

(Compute the minimal reverse test in the following manner. Define $N:=\sum
_{x}\left\vert \hat{\varphi}_{x}\right\rangle \left\langle e_{x}\right\vert $,
$D_{p}:=\sum_{x}p(x)\left\vert e_{x}\right\rangle \left\langle e_{x}%
\right\vert $, $D_{q}:=\sum_{x}q(x)\left\vert e_{x}\right\rangle \left\langle
e_{x}\right\vert $ . Then, $\hat{\sigma}=ND_{q}N^{\dagger}$. Therefore, there
is a unitary $U$ with \
\[
\hat{\sigma}^{1/2}=ND_{q}^{1/2}U.
\]
Therefore,
\[
\hat{\sigma}^{-1/2}\hat{\rho}\hat{\sigma}^{-1/2}=\sum_{x}r_{x}U^{\dagger
}\left\vert e_{x}\right\rangle \left\langle e_{x}\right\vert U.
\]
Therefore, the minimal reverse test maps $\delta_{x}$ to the constant multiple
of $\hat{\sigma}^{1/2}U^{\dagger}\left\vert e_{x}\right\rangle \left\langle
e_{x}\right\vert U\hat{\sigma}^{1/2}=q(x)\left\vert \hat{\varphi}%
_{x}\right\rangle \left\langle \hat{\varphi}_{x}\right\vert $. )

Therefore,%
\begin{align*}
\mathrm{D}_{f}(p\Vert q)  &  =\mathrm{D}_{f}^{\max}(\hat{\rho}\Vert\hat
{\sigma})=\mathrm{D}_{f}^{\max}\left(  \Lambda(\rho)\Vert\Lambda
(\sigma)\right) \\
&  \leq\mathrm{D}_{f}^{\max}(\rho\Vert\sigma)\leq\mathrm{D}_{f}(p\Vert q),
\end{align*}
indicating
\[
\mathrm{D}_{f}^{\max}\left(  \rho||\sigma\right)  =\mathrm{D}_{f}^{\max
}\left(  \hat{\rho}||\hat{\sigma}\right)  =\mathrm{D}_{f}\left(  p||q\right)
.
\]
By Theorem\thinspace\ref{th:unique}, the minimal reverse test of $\left\{
\rho,\sigma\right\}  $ should be essentially identical to $\left(
\Gamma,\left\{  p,q\right\}  \right)  $. But by the assumption (\ref{r-r}),
$\left(  \Gamma,\left\{  p,q\right\}  \right)  $ has to be the minimal reverse
test. Therefore, by Lemma\thinspace\ref{th:preserve}, (\ref{r-s->r-s}) implies
(\ref{phi->phi}).

\section{When one of the argument is a pure state}

\label{sec:pure}

From this section, again we remove the assumption of operator convexity and
$f(0)=0$, and come back to our initial assumption (FC). To start, we treat the
case where one of the argument is rank\thinspace-1.

Suppose $\sigma$ is rank\thinspace-1(the other case is reduce to this case by
replacing $f$ by $\hat{f}$), and apply (\ref{Dmax=}). Since $\tilde{\rho}$ is
constant multiple of $\sigma$, we have:
\begin{equation}
\mathrm{D}_{f}^{\max}\left(  \rho||\sigma\right)  =\sigma_{11}f\left(
\sigma_{11}^{-1}\tilde{\rho}\right)  +\hat{f}(0)\left(  \mathrm{tr}%
\,\rho-\tilde{\rho}\right)  ,\label{Dmax-example}%
\end{equation}
Though this coincide with (\ref{Dmax-formula}), it holds irrespective of the
assumption of operator convexity.

Especially, if $\rho$ is also rank\thinspace-1, $\tilde{\rho}=0$, thus%
\[
\mathrm{D}_{f}^{\max}\left(  \rho||\sigma\right)  =f(0)\mathrm{tr}%
\,\sigma+\hat{f}(0)\mathrm{tr}\,\rho,
\]
where $f(0)$ and/or $\hat{f}(0)$ may be $\infty$.

\section{Total variation distance}

\label{sec:total-variation}

\subsection{Set up and a general formula}

The divergence corresponding to $f(r)=\left\vert 1-r\right\vert $,%
\[
\mathrm{D}_{\left\vert 1-r\right\vert }(p\Vert q)=\left\Vert p-q\right\Vert
_{1},
\]
is called total variation distance. Its common quantum version is
\[
\left\Vert \rho-\sigma\right\Vert _{1}=\sup_{M}\left\Vert P_{\rho}%
^{M}-P_{\sigma}^{M}\right\Vert _{1},
\]
where $P_{\rho}^{M}$ is the distribution of the outcome of the measurement $M$
under $\rho$. This quantum version in fact is the smallest of all the quantum
versions satisfying (D1') and (D2):%
\begin{equation}
\mathrm{D}_{\left\vert 1-r\right\vert }^{Q}(\rho\Vert\sigma)\geq\left\Vert
\rho-\sigma\right\Vert _{1}. \label{Df1min}%
\end{equation}
Observe (D1') and (D2) imply
\[
\mathrm{D}_{\left\vert 1-r\right\vert }^{Q}(\rho\Vert\sigma)\geq\left\Vert
P_{\rho}^{M}-P_{\sigma}^{M}\right\Vert _{1}.
\]
Maximization of the RHS about $M$ leads to (\ref{Df1min}).

In this section we study $\mathrm{D}_{\left\vert 1-r\right\vert }^{\max}%
(\rho\Vert\sigma)$. Given a reverse test $\left(  \Gamma,\left\{  p,q\right\}
\right)  $ of $\left\{  \rho,\sigma\right\}  $, we define $\left(
\Gamma^{\prime},\left\{  p^{\prime},q^{\prime}\right\}  \right)  $, where
$\left\{  p^{\prime},q^{\prime}\right\}  $ are probability distributions on
$\left\{  0,1,2\right\}  $:
\begin{align}
\Gamma^{\prime}\left(  \delta_{0}\right)   &  :=\frac{1}{\mathrm{tr}%
\,A}A,\Gamma^{\prime}\left(  \delta_{1}\right)  :=\frac{\rho-A}{\mathrm{tr}%
\,(\rho-A)},\Gamma^{\prime}\left(  \delta_{2}\right)  :=\frac{\sigma
-A}{\mathrm{tr}\,(\sigma-A)},\nonumber\\
p^{\prime}(0) &  :=\mathrm{tr}\,A,\,\,p^{\prime}(1):=\mathrm{tr}%
\,(\rho-A),\,\,\,p^{\prime}\left(  2\right)  :=0,\nonumber\\
q^{\prime}(0) &  :=\mathrm{tr}\,A,\,\,q^{\prime}(1):=0,\,\,\,q^{\prime}\left(
2\right)  :=\mathrm{tr}\,(\sigma-A).\label{opt-|1-l|}%
\end{align}
where
\[
A:=\sum_{x\in\mathcal{X}}\min\left\{  p(x),q(x)\right\}  \Gamma(\delta_{x}).
\]
Then $\left(  \Gamma^{\prime},\left\{  p^{\prime},q^{\prime}\right\}  \right)
$ is a reverse test of $\left\{  \rho,\sigma\right\}  $ with $\left\Vert
p^{\prime}-q^{\prime}\right\Vert _{1}=\left\Vert p-q\right\Vert _{1}$.
Intuitively, $\Gamma^{\prime}\left(  \delta_{0}\right)  $ takes care of the
common part of two states, and $\Gamma^{\prime}\left(  \delta_{1}\right)  $
and $\Gamma^{\prime}\left(  \delta_{2}\right)  $ compensates the remainder.

Therefore, without loss of generality, we may restrict reverse tests to those
in the form of (\ref{opt-|1-l|}). Therefore:
\begin{equation}
\mathrm{D}_{\left\vert 1-r\right\vert }^{\max}(\rho\Vert\sigma)=\inf\left\{
\mathrm{tr}\,\left(  \rho+\sigma-2A\right)  ;A\geq0,\rho\geq A,\sigma\geq
A\right\}  .\label{Dmax-|1-l|}%
\end{equation}

\subsection{Invertible reverse test}

In this subsection and the next, under the assumption that $\mathrm{tr}%
\,\rho=\mathrm{tr}\,\sigma=1$, we study conditions for
\begin{equation}
\mathrm{D}_{\left\vert 1-r\right\vert }^{\max}(\rho\Vert\sigma)=\Vert
\rho-\sigma\Vert_{1}.\label{D=TV}%
\end{equation}
This identity implies uniqueness of quantum version of statistical distance,
and also indicates that classical total variation distance embedded into
quantum states can be completely recovered by measurements. It turns out that
the size of the set of all $\left\{  \rho,\sigma\right\}  $ 's satisfying
(\ref{D=TV}) is substantial. This is in contrast with the case of operator
convex functions, where the equivalence of (\ref{D=TV}) holds almost
exclusively for commutative pairs of states (see Subsection\thinspace
\ref{subsec:invertible}). 

Dropping the constraint $A\geq0$,
\begin{align*}
\mathrm{D}_{\left\vert 1-r\right\vert }^{\max}(\rho\Vert\sigma) &  \geq
\inf\left\{  \mathrm{tr}\,\left(  \rho+\sigma-2A\right)  ;\rho\geq
A,\sigma\geq A\right\}  \\
&  =\inf\left\{  2\mathrm{tr}\,(\rho-A);\rho\geq A,\sigma\geq A\right\}  \\
&  =\inf\left\{  2\mathrm{tr}\,(\rho-A);\rho-A\geq0,\rho-A\geq\rho
-\sigma\right\}  \\
&  =2\mathrm{tr}\,\left[  \rho-\sigma\right]  _{+}=\Vert\rho-\sigma\Vert_{1}.
\end{align*}
Here, the minimum in the third line is achieved if $\rho-A=\,\left[
\rho-\sigma\right]  _{+}$. ($\left[  X\right]  _{+}$ is the positive part of
the self-adjoint operator $X$.)

Therefore, (\ref{D=TV}) holds iff \
\begin{equation}
A=\rho-\,\left[  \rho-\sigma\right]  _{+}=\frac{1}{2}\left(  \rho
+\sigma-\left\vert \rho-\sigma\right\vert \right)  \geq0. \label{A>0}%
\end{equation}
(Here, $\left\vert X\right\vert :=\sqrt{X^{\dagger}X}$.) \ Another necessary
and sufficient condition is the existence of $A$, $\Delta_{1}$, $\Delta
_{2}\geq0$ with
\begin{align}
\rho &  =A+\Delta_{1},\sigma=A+\Delta_{2},\,\label{A+d}\\
\Delta_{1}\Delta_{2}  &  =0. \label{dd=0}%
\end{align}
To see this, observe
\begin{align*}
\Vert\Delta_{1}-\Delta_{2}\Vert_{1}  &  =\Vert\rho-\sigma\Vert_{1}\\
&  \leq\mathrm{D}_{\left\vert 1-r\right\vert }^{\max}(\rho\Vert\sigma
)=\min\left\{  \mathrm{tr}\,\Delta_{1}+\mathrm{tr}\,\Delta_{2};(\ref{A+d}%
),\Delta_{1}\geq0,\Delta_{2}\geq0\right\}  .
\end{align*}
For (\ref{D=TV}) to hold, existence of $\Delta_{1}$, $\Delta_{2}$ with
$\mathrm{tr}\,\Delta_{1}+\mathrm{tr}\,\Delta_{2}=\Vert\Delta_{1}-\Delta
_{2}\Vert_{1}$ is necessary and sufficient. Thus $\Delta_{1}\Delta_{2}=0$.

Of course, in general, (\ref{A>0}) is not true. For example, if $\rho$ is a
pure state and $\rho\neq c\sigma$, it is not true. (Let $f=\left\vert
1-r\right\vert $ in the formula (\ref{Dmax-example}). Then what we obtain is
very much different from $\Vert\rho-\sigma\Vert_{1}$.) However, if $\rho$ and
$\sigma$ are very close so that
\begin{equation}
\Vert\,\left\vert \rho-\sigma\right\vert \Vert\leq\text{minimum eigenvalue of
}\rho+\sigma, \label{close}%
\end{equation}
it is true.

Another sufficient condition is
\[
\left(  \rho-\sigma\right)  ^{2}=\left\vert \rho-\sigma\right\vert ^{2}%
\leq\left(  \rho+\sigma\right)  ^{2}.
\]
To see this is sufficient, take the square root of both sides of inequality:
then we obtain (\ref{A>0}). (Recall $\sqrt{\cdot}$ is operator monotone. This
condition is not necessary, since $r^{2}$ is not operator monotone.)
Rearranging the terms, we have
\begin{equation}
\rho\sigma+\sigma\rho\geq0. \label{ls+sl}%
\end{equation}

By (\ref{close}), $\mathrm{D}_{\left\vert 1-r\right\vert }^{\max}\left(
\rho\Vert\rho+\varepsilon X\right)  =\left\Vert \rho-\left(  \rho+\varepsilon
X\right)  \right\Vert _{1}$ for all small $\varepsilon>0$ and for all $X$. On
the other hand, if $f$ is operator convex, Proposition\thinspace
\ref{prop:invertible-2} indicates that $\mathrm{D}_{f}^{\max}\left(  \rho
\Vert\rho+\varepsilon X\right)  \neq\mathrm{D}_{f}^{\min}\left(  \rho\Vert
\rho+\varepsilon X\right)  $ for most of small $\varepsilon>0$ unless $\rho$
and $X$ commute.

\subsection{2\thinspace-\thinspace dimensional case}

In this subsection, we assume $\dim\mathcal{H}=2$ and $\mathrm{tr}%
\,\rho=\mathrm{tr}\,\sigma=1$, and compute the set $\left\{  \sigma
;\text{(\ref{D=TV})}\right\}  $ for each fixed $\rho$, using the necessary and
sufficient condition given by (\ref{A+d}) and (\ref{dd=0}). As it turns out,
this set is the spheroid, with focal points $\rho$ and $\mathbf{1}-\rho$, and
touching to the surface of Bloch sphere at each end of the longest axis.

Since $\mathrm{tr}\,\rho=\mathrm{tr}\,\sigma=1$,
\[
c:=\mathrm{tr}\,\Delta_{1}=\mathrm{tr}\,\Delta_{2}=1-\mathrm{tr}\,A,
\]
and
\[
0\leq c\leq1.
\]
Let $v_{\rho}$ , $v_{\sigma}$, $u_{1}$, $u_{2}$, and $u_{A}$ be the Bloch
vector of $\rho$, $\sigma$, $\frac{1}{c}\Delta_{1}$, $\frac{1}{c}\Delta_{2}$ ,
and $\frac{1}{1-c}A$ , respectively. Also, (\ref{dd=0}) holds iff $\Delta_{1}$
and $\Delta_{2}$ are rank\thinspace-\thinspace1 and $u_{2}=-u_{1}.$Therefore,
by (\ref{A+d}),
\[
v_{\rho}=cu_{1}+\left(  1-c\right)  u_{A},\,v_{\sigma}=-cu_{1}+\left(
1-c\right)  u_{A}.
\]
Therefore,%
\[
v_{\sigma}-v_{\rho}=-2cu_{1},\,v_{\sigma}-\left(  -v_{\rho}\right)  =2\left(
1-c\right)  u_{A}.
\]
Let $\left\Vert \cdot\right\Vert $ denote the Euclid norm in $\mathbb{R}^{3}$,
and
\[
\left\Vert v_{\sigma}-v_{\rho}\right\Vert +\left\Vert v_{\sigma}-\left(
-v_{\rho}\right)  \right\Vert =2\left(  c\left\Vert u_{1}\right\Vert +\left(
1-c\right)  \left\Vert u_{A}\right\Vert \right)  \leq2.
\]

The set $\left\{  \sigma;\text{(\ref{D=TV})}\right\}  $ is fairly large. For
example, if the largest eigenvalue of $\rho$ is $\leq0.85$, this occupies more
than the half of the volume of the Bloch sphere.

If
\[
\rho=\left[
\begin{array}
[c]{cc}%
a & \overline{c}\\
c & b
\end{array}
\right]  ,\sigma=\left[
\begin{array}
[c]{cc}%
a & -\overline{c}\\
-c & b
\end{array}
\right]  ,\,\,\,\left(  a\geq b\right)
\]
the minimization problem (\ref{Dmax-|1-l|}) is solved explicitly. With
$Z:=\mathrm{diag}\left(  1,-1\right)  $, $\sigma=Z\rho Z^{\dagger}$,
$\rho=Z\sigma Z^{\dagger}$. Thus, if $A$ satisfies constrains of
(\ref{Dmax-|1-l|}), so does $\frac{1}{2}(ZAZ^{\dagger}+A)$, and $\mathrm{tr}%
\,A=\mathrm{tr}\,\frac{1}{2}(ZAZ^{\dagger}+A)$. Therefore, without loss of
generality, we suppose $A$ is diagonal. After some elementary analysis, the
optimal $A$ turns out to be
\[
A=\left\{
\begin{array}
[c]{cc}%
\mathrm{diag}\,(a-\left\vert c\right\vert ,b-\left\vert c\right\vert ), &
(a\geq b\geq\left\vert c\right\vert )\\
\mathrm{diag}(a-\left\vert c\right\vert ^{2}b^{-1},0), & (a\geq\left\vert
c\right\vert \geq b)
\end{array}
\right.
\]
and we have
\[
\mathrm{D}_{\left\vert 1-r\right\vert }^{\max}(\rho\Vert\sigma)=\left\{
\begin{array}
[c]{cc}%
4|c|=\Vert\rho-\sigma\Vert_{1}, & (a\geq b\geq\left\vert c\right\vert )\\
2(b+\left\vert c\right\vert ^{2}b^{-1}). & (a\geq\left\vert c\right\vert \geq
b)
\end{array}
\right.
\]

\section{Dual Representation and Continuity}

\label{sec:duality}

In this section we give the dual of (\ref{Dmax-def}), or representation of
$\mathrm{D}_{f}^{\max}$ by maximization of a linear functional:%

\begin{equation}
\mathrm{D}_{f}^{\max}(\rho\Vert\sigma)=\sup_{\left(  W_{1},W_{2}\right)
\in\mathcal{W}_{f}^{\max}\left(  \mathcal{H}\right)  }\left\{  \mathrm{tr}%
\,\rho W_{1}+\mathrm{tr}\,\sigma W_{2}\right\}  , \label{Dmax-dual}%
\end{equation}
where
\begin{align*}
\mathcal{W}_{f}^{\max}\left(  \mathcal{H}\right)   &  :=\left\{  \left(
W_{1},W_{2}\right)  ;sW_{1}+tW_{2}-g_{f}(s,t)\mathbf{1}_{\mathcal{H}}%
\leq0,\forall s,t\in\lbrack0,1]\right\} \\
&  =\left\{  \left(  W_{1},W_{2}\right)  ;f(r)\mathbf{1}-rW_{1}-W_{2}%
\geq0,\,r\geq0\right\}
\end{align*}
(To see the equality in the second line, recall $g_{f}$ is positively
homogeneous and continuous.).

Let $\mathrm{D}_{f}^{Q}$ be a lower semicontinuous, proper, positively
homogeneous, and convex function over positive operators. Then by Corollary
13.5.1 of \cite{Rockafellar}, it should be in the form of
\[
\mathrm{D}_{f}^{Q}(\rho\Vert\sigma)=\sup_{\left(  W_{1},W_{2}\right)
\in\mathcal{W}_{f}^{Q}}\mathrm{tr}\,(\rho W_{1}+\sigma W_{2}),
\]
where $\mathcal{W}_{f}^{Q}$ is, without loss of generality, convex and
unbounded from below. In addition, suppose $\mathrm{D}_{f}^{Q}$ satisfies
(D1') and (D2).

Since $\mathrm{D}_{f}^{Q}$ satisfies (D2) with $p(1)=s$, $q(1)=t$,
$p(x)=q(x)=0$, $(x\neq1)$,
\begin{align*}
&  \mathrm{D}_{f}^{Q}\left(  s\left\vert e_{1}\right\rangle \left\langle
e_{1}\right\vert \Vert\,t\left\vert e_{1}\right\rangle \left\langle
e_{1}\right\vert \right) \\
&  =\sup_{\left(  W_{1},W_{2}\right)  \in\mathcal{W}_{f}^{Q}}\left(
s\left\langle e_{1}\right\vert W_{1}\left\vert e_{1}\right\rangle
+t\left\langle e_{1}\right\vert W_{2}\left\vert e_{1}\right\rangle \right)
=g_{f}(s,t)\text{.}%
\end{align*}
Since the second identity is true for all $s\geq0$, $t\geq0$ and $\left\vert
e_{1}\right\rangle $ with $\left\Vert e_{1}\right\Vert =1$, $\left(
W_{1},W_{2}\right)  \in\mathcal{W}_{f}^{\max}$. Therefore, $\mathcal{W}%
_{f}^{Q}\subset\mathcal{W}_{f}^{\max}$. Also, the RHS of (\ref{Dmax-dual})
satisfies (D2).

Therefore, the RHS of (\ref{Dmax-dual}) is the largest of all proper,
positively homogeneous, convex, lower and semicontinuous functionals with
(D2). As is easily verified, it also satisfies (D1) and (D1').

On the other hand, $\mathrm{D}_{f}^{\max}$ is the largest of all functionals
with (D1') and (D2) (Lemma\thinspace\ref{lem:Dmax>DQ}) and turns out to be
proper, positively homogeneous, convex, lower semicontinuous. Therefore:
\begin{equation}
\mathrm{cl\,D}_{f}^{\max}(\rho\Vert\sigma)=\sup_{\left(  W_{1},W_{2}\right)
\in\mathcal{W}_{f}^{\max}\left(  \mathcal{H}\right)  }\left\{  \mathrm{tr}%
\,\rho W_{1}+\mathrm{tr}\,\sigma W_{2}\right\}  . \label{clDmax=}%
\end{equation}

\begin{remark}
From above argument, if $f$ satisfies (FC), $\mathrm{D}_{f}^{\max}$ is the
largest of all $\mathrm{D}_{f}^{Q}$'s which are lower semicontinuous, proper,
positively homogeneous, convex, and satisfy (D2).
\end{remark}

Below, we show (\ref{Dmax-dual}) by proving that $\mathrm{D}_{f}^{\max}$ is
lower semicontinuous. It suffices to show this on
\[
\mathfrak{D}:=\{\left(  \rho,\sigma\right)  ;\rho\geq0,\sigma\geq
0,\,\mathrm{tr}\,\rho\leq1,\mathrm{tr}\,\sigma\leq1\},
\]
by $\mathrm{D}_{f}^{\max}(\lambda\rho\Vert\lambda\sigma)=\lambda\mathrm{D}%
_{f}^{\max}(\rho\Vert\sigma)$, $\forall\lambda>0$.Thus we suppose each reverse
test $(\Gamma,\{p,q\})$ satisfies%
\[
\sum_{x\in\mathcal{X}}q(x)\leq1,\sum_{x\in\mathcal{X}}p(x)\leq1.
\]
Also, by Lemma\thinspace\ref{lem:caratheodory}, we suppose $\{p,q\}$ are over
$\mathcal{X}$ with $\left\vert \mathcal{X}\right\vert \leq\left(
\dim\mathcal{H}\right)  ^{2}+3$. The set of all such reverse tests
$\mathfrak{T}=\{\tau=(\Gamma,\{p,q\})\}$ can be identified with a compact
subset of a finite dimensional real vector space.

Define maps $F_{1}(\tau):=\{\Gamma(p),\Gamma(q)\}$ and $F_{2}(\tau
):=\mathrm{D}_{f}(p\Vert q)$, and let
\[
\mathfrak{U}:=\{\upsilon=(\upsilon^{1},\upsilon^{2});\upsilon^{1}=F_{1}%
(\tau),\upsilon^{2}\geq F_{2}(\tau),\tau\in\mathfrak{T}\}.
\]

\begin{lemma}
Suppose (FC) is satisfied. Then the set $\mathfrak{U}$ is closed and identical
to $\mathrm{epi\,D}_{f}^{\max}|_{\mathfrak{D}}$.
\end{lemma}

\begin{proof}
Suppose $(\upsilon^{1},\upsilon^{2})\in\mathrm{cl}\,\mathfrak{U}$. Then for
any $\varepsilon>0$, $\overline{B}_{\varepsilon}(\upsilon^{1})\times
C_{\varepsilon}(\upsilon^{2})\cap\mathfrak{U}$ is not empty, where
$\overline{B}_{\varepsilon}(\upsilon^{1})$ is the closed $\varepsilon$-ball
centered at $\upsilon^{1}$, and $C_{\varepsilon}(\upsilon^{2}):=\{t;t\leq
\upsilon^{2}+\varepsilon\}$. Therefore, all the sets in the family
\[
\{F_{1}^{-1}(\overline{B}_{\varepsilon}(\upsilon^{1}))\cap F_{2}%
^{-1}(C_{\varepsilon}(\upsilon^{2}))\cap\mathfrak{T\}}_{\varepsilon>0}%
\]
are not empty, and in fact, they are closed subsets of the compact set
$\mathfrak{T}$, since $F_{1}$ is continuous and $F_{2}$ is lower
semicontinuous. Since the family has finite intersection property, the
intersection of these sets is not empty. Any element $\tau$ of this
intersection satisfies $\upsilon^{1}=F_{1}(\tau)$ and $\upsilon^{2}=F_{2}%
(\tau)$, indicating $\upsilon\in\mathfrak{U}$. Therefore, $\mathfrak{U}$ is closed.

The second statement follows by
\[
\mathfrak{U\subset\,}\mathrm{epi\,D}_{f}^{\max}|_{\mathfrak{D}}%
\mathfrak{\subset\,}\mathrm{cl\,}\mathfrak{U}.
\]
The first "$\mathfrak{\subset}$" by $F_{2}(\tau)\geq\mathrm{D}_{f}^{\max
}(\Gamma(p)\Vert\Gamma(q))$, and the second one is by the definition
(\ref{Dmax-def}) of $\mathrm{D}_{f}^{\max}$.
\end{proof}

\begin{lemma}
\label{lem:lower-semi}Suppose (FC) is satisfied. Then, $\mathrm{D}_{f}^{\max}$
is lower semicontinuous. Moreover, for each $\{\rho,\sigma\}$ such that
$\mathrm{D}_{f}^{\max}\dot{(}\rho\Vert\sigma)<\infty$, the infimum in
(\ref{Dmax-def}) is achieved by some $(\Gamma,\{p,q\})$.
\end{lemma}

\begin{proof}
It suffices to prove the assertion on $\mathfrak{D}$. The closedness of
$\mathrm{D}_{f}^{\max}|_{\mathfrak{D}}$ follows by closedness of
$\mathrm{epi\,D}_{f}^{\max}|_{\mathfrak{D}}$. \ Also, by the previous lemma,
\begin{align*}
\mathrm{\,epi\,D}_{f}^{\max}|_{\mathfrak{D}} &  =\{\,(\rho,\sigma
,t)\,;\text{\thinspace}(\rho,\sigma)\in\mathfrak{D},\,t\geq\mathrm{D}%
_{f}^{\max}(\rho\Vert\sigma)\,\}\\
&  =\{\,(\Gamma(p),\Gamma(q),t)\,;\,\tau\in\mathfrak{T},\,t\geq\mathrm{D}%
_{f}(p\Vert q)\}.
\end{align*}
Therefore, to each $\{\rho,\sigma\}$ there is a reverse test $(\Gamma
,\{p,q\})$ of $\{\rho,\sigma\}$ with $\{\,t\,;t\geq\mathrm{D}_{f}^{\max}%
(\rho\Vert\sigma)\,\}=\left\{  \,t\,;\,t\geq\mathrm{D}_{f}(p\Vert q)\right\}
$, or equivalently, $\mathrm{D}_{f}^{\max}(\rho\Vert\sigma)=\mathrm{D}%
_{f}(p\Vert q)$. This reverse test achieves the infimum in (\ref{Dmax-def}).
\end{proof}

By this lemma and (\ref{clDmax=}):

\begin{theorem}
\label{th:dual}If $f$ satisfies (FC), then (\ref{Dmax-dual}) holds. Moreover,
for each $\{\rho,\sigma\}$ such that $\mathrm{D}_{f}^{\max}\dot{(}\rho
\Vert\sigma)<\infty$, the infimum in (\ref{Dmax-def}) is achieved by some
$(\Gamma,\{p,q\})$.
\end{theorem}

\subsection{When $f$ is operator convex}

When $f$ satisfies the condition (F) and $\rho>0$ and $\sigma>0$, we can write
$\left(  W_{1\ast},W_{2\ast}\right)  $ achieving the maximum in
(\ref{Dmax-dual}) explicitly. \ 

Since $f$ is operator convex, it is differentiable. Hence the Frechet
derivative $\mathrm{D}f\left(  T\right)  $ of $f$ i.e., a linear transform in
$\mathcal{B}\left(  \mathcal{H}\right)  $ with%
\[
\Vert f\left(  T+X\right)  -f\left(  T\right)  -\mathrm{D}f\left(  T\right)
\left(  X\right)  \Vert_{2}=o\left(  \Vert X\Vert_{2}\right)
\]
is given by, in the basis which diagonalizes $T$,
\begin{equation}
\,\mathrm{D}f\left(  T\right)  \left(  X\right)  =\left[  f^{\left[  1\right]
}\left(  t_{i},t_{j}\right)  X_{i,j}\right]  ,\label{Frechet}%
\end{equation}
where $t_{i}$ ($i=1,\cdots$) are eigenvalues of $T$, and
\[
f^{\left[  1\right]  }\left(  t,t^{\prime}\right)  :=\left\{
\begin{array}
[c]{cc}%
\frac{f\left(  t\right)  -f\left(  t^{\prime}\right)  }{t-t^{\prime}}, &
\left(  t\neq t^{\prime}\right)  ,\\
f^{\prime}\left(  t\right)  , & \left(  t=t^{\prime}\right)  .
\end{array}
\right.
\]

Using $d(\rho,\sigma)$ as of (\ref{d-def}),%

\begin{align*}
\left.  \frac{\mathrm{d}}{\mathrm{d}t}\mathrm{D}_{f}^{\max}(\rho+tX\Vert
\sigma)\right\vert _{t=0} &  =\left.  \frac{\mathrm{d}}{\mathrm{d}%
t}\mathrm{tr}\,\,\sigma f(\sigma^{-\frac{1}{2}}\left(  \rho+tX\right)
\sigma^{-\frac{1}{2}})\right\vert _{t=0}\\
&  =\mathrm{tr}\,\sigma\mathrm{D}f(d(\rho,\sigma))(\sigma^{-\frac{1}{2}%
}X\sigma^{-\frac{1}{2}})\\
&  =\mathrm{tr}\,X\sigma^{-\frac{1}{2}}\mathrm{D}f(d(\rho,\sigma
))(\sigma)\sigma^{-\frac{1}{2}},
\end{align*}
where the last identity is by self-adjointness of $\mathrm{D}f\left(
T\right)  \left(  \cdot\right)  $ with respect to the inner product
$\mathrm{tr}\,XY$,
\begin{align}
\mathrm{tr}\,Y\mathrm{D}f\left(  T\right)  \left(  X\right)   &  =\sum
_{i,j}\overline{\rho_{2,i,j}}f^{[1]}(t_{i},t_{j})X_{i,j}=\sum_{i,j}%
\overline{f^{[1]}(t_{i},t_{j})\rho_{2,i,j}}X_{i,j}\nonumber\\
&  =\mathrm{tr}\,X\mathrm{D}f\left(  T\right)  \left(  Y\right)
.\label{self-adjoint}%
\end{align}

Replacing $f$ by $\hat{f}$, the derivative about the second argument is
computed similarly. Therefore,
\[
W_{1\ast}=\sigma^{-\frac{1}{2}}\left\{  \mathrm{D}\,f(d(\rho,\sigma
))(\sigma)\right\}  \sigma^{-\frac{1}{2}},\,\,W_{2\ast}=\rho^{-\frac{1}{2}%
}\left\{  \mathrm{D}\hat{f}(d(\sigma,\rho))(\rho)\right\}  \rho^{-\frac{1}{2}%
}.
\]
achieves the maximum in (\ref{Dmax-dual}). In fact, by (\ref{self-adjoint}),
\begin{align*}
\mathrm{tr}\,\left(  \rho W_{1,\ast}+\sigma W_{2,\ast}\right)   &
=\mathrm{tr}\,\sigma\mathrm{D}f(d(\rho,\sigma))(d(\rho,\sigma))+\mathrm{tr}%
\,\rho\mathrm{D}\hat{f}(d(\sigma,\rho))(d(\sigma,\rho))\\
&  =\mathrm{tr}\,\sigma f^{\prime}(d(\rho,\sigma))d(\rho,\sigma)+\mathrm{tr}%
\,\rho\hat{f}^{\prime}(d(\sigma,\rho))d(\sigma,\rho)\\
&  =\mathrm{D}_{f}^{\max}(\rho\Vert\sigma).
\end{align*}

For example, if $f(r)=r^{2}$,  $\hat{f}(r)=1/r$, $\mathrm{D}f\left(  T\right)
\left(  X\right)  =TX+XT$ and $\mathrm{D}\hat{f}\left(  T\right)  \left(
X\right)  =-TXT^{-1}$,
\[
W_{1\ast}=\sigma^{-1}\rho+\rho\sigma^{-1},W_{2\ast}=-\sigma^{-1}\rho^{2}%
\sigma^{-1}.
\]

\subsection{On continuity}

\label{sec:continuity}

In this subsection, some remarks on continuity of $\mathrm{D}_{f}^{\max}$ are
in order. By Lemma\thinspace\ref{lem:lower-semi} and Proposition\thinspace
\ref{prop:continuous}, if $f$ satisfies (FC),
\begin{equation}
\lim_{\varepsilon\downarrow0}\mathrm{D}_{f}^{\max}(\rho_{\varepsilon}%
\Vert\sigma_{\varepsilon})=\mathrm{D}_{f}^{\max}(\rho\Vert\sigma),
\label{limDf=Df}%
\end{equation}
where $\left\{  \left(  \rho_{\varepsilon},\sigma_{\varepsilon}\right)
\right\}  _{\varepsilon>0}$ is a straight line in the effective domain of
$\mathrm{D}_{f}^{\max}$.

$\left\{  \left(  \rho_{\varepsilon},\sigma_{\varepsilon}\right)  \right\}
_{\varepsilon>0}$ cannot be arbitrary curve for (\ref{limDf=Df}) to hold. To
see this, suppose that $\hat{f}(0)<\infty$ and $\sigma$ is a pure state, and
use (\ref{Dmax-example}). Let%
\[
\sigma=\left[
\begin{array}
[c]{cc}%
a & 0\\
0 & 0
\end{array}
\right]  ,\,\rho_{\varepsilon}=\left[
\begin{array}
[c]{cc}%
b & \sqrt{\varepsilon}C^{\dagger}\\
\sqrt{\varepsilon}C & \varepsilon D
\end{array}
\right]  ,
\]
and $\rho_{1}\geq0$, $\mathrm{tr}\,\rho_{\varepsilon}=1$. Then
\[
\tilde{\rho}_{\varepsilon}=b-\sqrt{\varepsilon}C(\varepsilon D)^{-1}%
\sqrt{\varepsilon}C^{\dagger}=b-CD^{-1}C^{\dagger}=\tilde{\rho}_{1}%
\]
is constant of $\varepsilon$, and
\[
\lim_{\varepsilon\downarrow0}\mathrm{D}_{f}^{\max}(\rho_{\varepsilon}%
\Vert\sigma)=\mathrm{D}_{f}^{\max}(\tilde{\rho}_{1}\Vert\sigma)+\hat
{f}(0)(1-\,\tilde{\rho}_{1})\neq\mathrm{D}_{f}^{\max}(\rho_{0}\Vert\sigma).
\]

However, $\{(\rho_{\varepsilon},\sigma_{\varepsilon})\}_{\varepsilon>0}$ need
not to be straight line, either. For example, consider a continuous curve
$\{\sigma_{\varepsilon}\}_{\varepsilon\geq0}$ of positive operators with
$\sigma_{0}=\sigma$, $\mathrm{supp}\sigma_{\varepsilon}\supset\mathrm{supp}%
\rho,\mathrm{supp}\sigma,$ and $\sigma_{\varepsilon}>\sigma$. Then, if
$f(0)=0$, (\ref{D'-monotone}) we have
\[
\lim_{\varepsilon\downarrow0}\mathrm{D}_{f}^{\max}(\rho\Vert\sigma
_{\varepsilon})\leq\mathrm{D}_{f}^{\max}(\rho\Vert\sigma).
\]
The opposite inequality results from lower semicontinuity of $\mathrm{D}%
_{f}^{\max}$. Thus (\ref{limDf=Df}) holds.

\subsection{Infinite dimensional separable Hilbert space}

\label{sec:infinite}

So far, we had supposed that the dimension of the underlying Hilbert space is
finite. Some of them, namely Sections\thinspace\ref{sec:property} and
\ref{sec:pure}, are obviously generalized to separable infinite dimensional
case. In this section, we consider generalization of (\ref{Dmax-dual}), thus
proving lower \thinspace semicontinuity. Also the existence of the minimum is
discussed. Throughout the section, we suppose $\hat{f}(0)<\infty$ and
$f(0)<\infty$. In this case, $g_{f}$ is not only lower semicontinuous, but
also continuous.

\begin{remark}
To see that $g_{f}$ is continuous, we only have to check it at the origin.
Observe $g_{f}(s,t)=\tilde{g}_{f}(s,t)+\hat{f}(0)s+f(0)\,t$, with $\tilde
{g}_{f}(s,t)\leq0$. Therefore, if $\left(  s_{k},t_{k}\right)  \rightarrow
(0,0)$,
\begin{align*}
\varlimsup_{k\rightarrow\infty}g_{f}(s_{k},t_{k})  &  =\varlimsup
_{k\rightarrow\infty}\{\tilde{g}_{f}(s,t)+\hat{f}(0)s_{k}+f(0)\,t_{k}%
\}=\varlimsup_{k\rightarrow\infty}\tilde{g}_{f}(s,t)\\
&  \leq0=g_{f}(s,t)\leq\varliminf_{k\rightarrow\infty}g_{f}(s_{k},t_{k}).
\end{align*}
indicating $\lim_{k\rightarrow\infty}$ $g_{f}(s_{k},t_{k})=g_{f}(s,t)$.
\end{remark}

In the RHS of (\ref{Dmax-dual}), let $\rho$ and $\sigma$ are positive element
of $\mathcal{B}_{1,sa}$, the space of all the self\thinspace-\thinspace
adjoint trace class operators (operators with finite trace), and $W_{1}$ and
$W_{2}$ are self\thinspace-\thinspace adjoint bounded operators,
$\mathcal{B}_{sa}$.

Also, we modify the definition of reverse tests. admitting all the positive
regular measures as input. To state that object, operator valued functions and
their integrals are used, see Section\thinspace2.3, \cite{Ryan} for these
concepts. In this new definition, the reverse test is specified by a regular
finite measure $\nu$ over the Borel sets of $[0,1]^{\times2}$, \ a $\nu
$-\thinspace measurable function $Z(s,t)$ from $(s,t)\in$ $[0,1]^{\times2}$
into $\mathcal{B}_{1,sa}$ with
\begin{align}
\mathrm{tr}\,Z(s,t)  &  =1,\nu\text{-\thinspace a.e.,}\label{trZ=1}\\
\int s\,Z(s,t)\mathrm{d}\nu &  =\rho,\int tZ(s,t)\mathrm{d}\nu=\sigma,
\label{reverse-test-infty}%
\end{align}
where the integrals of operator valued functions are understood as short for
\[
\int s\mathrm{tr}\,\,Z(s,t)W\mathrm{d}\nu=\mathrm{tr}\,\rho W,\,\forall
W\in\mathcal{B}_{sa}.
\]
Thus our new definition is:
\begin{equation}
\mathrm{D}_{f}^{\max}(\rho\Vert\sigma):=\inf\left\{  \int g_{f}%
(s,t)\mathrm{tr}\,Z(s,t)\mathrm{d}\nu;\text{(\ref{trZ=1}) and
(\ref{reverse-test-infty})}\right\}  \label{reverse-test-3}%
\end{equation}

This defines a positive map from measures which are absolutely continuous
relative to $\nu$ into $\mathcal{B}_{1,sa}$ such that $\Gamma\left(
\mu\right)  =\int\frac{\mathrm{d\,}\mu}{\mathrm{d\,}\nu}\,Z(s,t)\mathrm{d}\nu
$. This is `trace preserving' in the sense $\mu([0,1]^{\times2})=\mathrm{tr}%
\,\Gamma\left(  \mu\right)  .$Thus the pair $\nu$ and $Z$ represents a
"reverse test" $\left(  \Gamma,\left\{  P,Q\right\}  \right)  $, where $P$ and
$Q$ are positive measures with density $s$ and $t$, respectively.

\begin{remark}
A function $Z(s,t)$ is $\nu$- measurable iff its norm\thinspace-\thinspace
approximable by simple functions, but the definition is equivalent to that the
scalar valued function $(s,t)\rightarrow\mathrm{tr}\,Z(s,t)W$ is $\nu
$-measurable (Proposition 2.15, \cite{Ryan}). (\ref{reverse-test-infty}) may
be understood in the "weak" sense as stated, but since $\Vert sZ(s,t)\Vert
_{1}\leq\mathrm{tr}\,\,Z(s,t)$ is $\nu$-integrable, also can be understood as
Bochner integral, or the limit of integral of simple functions in norm. $\,$
\end{remark}

\begin{proposition}
\label{prop:duality} (Theorem 8.6.1, \cite{Luenberger}) Let $F$ be a
real-valued convex function defined on a convex subset $\Omega$ of a vector
space $\mathfrak{X}$, and let $G$ be a convex mapping of $\mathfrak{X}$ into a
partially ordered normed space $\mathfrak{Z}$. Define
\[%
\mu
_{0}:=\sup\{F(\vec{W});\vec{W}\in\Omega,G(\vec{W})\leq0\}.
\]
Then for any $\zeta_{0}^{\ast}\geq0$,
\begin{equation}
\mu_{0}\leq F_{\ast}\left(  \zeta_{0}^{\ast}\right)  :=\sup\{F(\vec
{W})+\langle\,\zeta_{0}^{\ast},G(\vec{W})\,\rangle;\vec{W}\in\Omega\}.
\label{weak-dual}%
\end{equation}
Also, if there exists an $\vec{W}_{0}$ such that $G(\vec{W}_{0})<0$ and
$\mu_{0}$ is finite,
\begin{equation}
\mu_{0}=\min_{\zeta_{0}^{\ast}\geq0}F_{\ast}(\zeta_{0}^{\ast}).
\label{strong-duality}%
\end{equation}

\end{proposition}

Below, we apply this proposition considering the RHS of (\ref{Dmax-dual}) as
the primal problem, and obtain the reverse test as its dual problem. To
proceed, we need to introduce a proper mathematical framework.

Consider the space $\mathcal{C}$ of continuous real valued functions on the
compact set $[0,1]^{\times2}$ and the space $\mathcal{B}_{sa}$ of the space of
the bounded self\thinspace-\thinspace adjoint linear operators on the Hilbert
space $\mathcal{H}$. Endorse $\mathcal{C}$ and $\mathcal{B}_{sa}$ with the
norm $\Vert h\Vert:=\sup_{(s,t)\in\lbrack0,1]^{\times2}}|h(s,t)|$ and the
operator norm $\Vert W\Vert$, respectively. From these two spaces, we compose
the linear space
\[
\left\{  \sum_{i=1}^{n}h^{(i)}W^{(i)}\,;h^{(i)}\in\mathcal{C}\,,W^{(i)}%
\in\mathcal{B}_{sa}\right\}  ,
\]
and its completion with respect to the projective norm
\[
\left\Vert z\right\Vert _{\pi}\colon=\inf\left\{  \sum_{i=1}^{n}\Vert
h^{(i)}\Vert\Vert W^{(i)}\Vert\,;\,z=\sum_{i=1}^{n}h^{(i)}W^{(i)}\,\right\}
\]
is denoted by $\mathfrak{Z}$. In fact, $\mathfrak{Z}$ is the projective tensor
product $\mathcal{C}\hat{\otimes}_{\pi}\mathcal{B}_{sa}$ . (That $\Vert
\cdot\Vert_{\pi}$ is a norm and $\Vert hW\Vert_{\pi}=\Vert h\Vert\Vert W\Vert$
is known \cite{Ryan}.)

Then for each $z\in\mathfrak{Z}$, there exist bounded sequences $\left\{
h^{(i)}\right\}  $ and $\left\{  W^{(i)}\right\}  $ with $z=\sum_{i=1}%
^{\infty}h^{(i)}W^{(i)}$ and
\[
\Vert z\Vert_{\pi}=\inf\left\{  \sum_{i=1}^{\infty}\Vert h^{(i)}\Vert\Vert
W^{(i)}\Vert;\,z=\sum_{i=1}^{\infty}h^{(i)}W^{(i)}\,\right\}  .
\]
One can endorse the partial order $\geq$ in $\mathfrak{Z}$ by
\[
z\geq0\Leftrightarrow\sum_{i=1}^{\infty}h^{(i)}(s,t)W^{(i)}\geq0,\forall
(s,t)\in\lbrack0,1]^{\times2}.
\]
The strict inequality $z>0$ means that $z$ is an interior point of the cone
$\left\{  z^{\prime};z^{\prime}\geq0\right\}  $.

Any bounded linear functional $\zeta^{\ast}$ on $\mathfrak{Z}$ is the
linearization of bilinear form on $\mathcal{C}$ and $\mathcal{B}_{sa}$ (see
Section 2.2, \cite{Ryan}) :%

\[
\zeta^{\ast}\left(  hW\right)  =\zeta^{\ast}\left(  h\right)  \left(
W\right)  ,
\]
where $\zeta^{\ast}\left(  h\right)  \left(  \cdot\right)  $ and $\zeta^{\ast
}\left(  \cdot\right)  \left(  W\right)  $ is an element of $\mathcal{B}%
_{sa}^{\ast}$ and $\mathcal{C}^{\ast}$, respectively.

Below, $\mathfrak{Z}$ is the one defined as above, and
\begin{align*}
\mathfrak{X}  &  =\Omega:=\{\vec{W}=(W_{1},W_{2});W_{1},W_{2}\in
\mathcal{B}_{sa}\},\\
F(\vec{W})  &  :=\mathrm{tr}\,\rho W_{1}+\mathrm{tr}\,\sigma W_{2},
\end{align*}

\begin{lemma}
\label{lem:dual} Suppose $g_{f}$ is positive, bounded and continuous on
$\left[  0,1\right]  ^{\times2}$. Suppose $(s,t)\rightarrow\eta_{s,t}(W)$ is
$\nu$\thinspace-\thinspace measurable function on $\left[  0,1\right]
^{\times2}$ and $W\rightarrow\eta_{s,t}(W)$ is a linear functional with
\begin{equation}
\left\vert \eta_{s,t}(W)\right\vert \leq\Vert W\Vert,\nu\text{-\thinspace
a.e.,} \label{|eta|=1-2}%
\end{equation}
and
\begin{equation}
\int s\,\eta_{s,t}(W)\mathrm{d}\,\nu=\mathrm{tr}\,\rho W,\,\int t\,\eta
_{s,t}(W)\mathrm{d}\,\nu=\mathrm{tr}\,\sigma W,\,\forall W\in\mathcal{B}_{sa}.
\label{eta-constraint}%
\end{equation}
Then
\begin{equation}
\min_{\eta}\int g_{f}(s,t)\eta_{s,t}(\mathbf{1})\mathrm{d}\,\nu=\sup_{\left(
W_{1},W_{2}\right)  \in\mathcal{W}_{f}}\left(  \mathrm{tr}\,W_{1}%
\rho+\mathrm{tr}\,W_{2}\sigma\right)  . \label{dual-pre-pre}%
\end{equation}

\end{lemma}

\begin{proof}
We apply Proposition$\,$\ref{prop:duality} with
\[
G(\vec{W})\colon=g_{1}W_{1}+g_{2}W_{2}-g_{f}\mathbf{1},
\]
where $g_{1}(s,t)\colon=s$, $g_{2}(s,t)\colon=t$. With $\zeta^{\ast}%
\in\mathfrak{Z}^{\ast}$, $\zeta^{\ast}\geq0$,
\begin{align*}
&  F_{\ast}\left(  \zeta^{\ast}\right)  =\sup_{\vec{W}}\{\mathrm{tr}\,\rho
W_{1}+\mathrm{tr}\,\sigma W_{2}-\zeta^{\ast}(g_{1}W_{1}+g_{2}W_{2}%
-g_{f}\mathbf{1})\}\\
&  =\sup_{\vec{W}}\left\{  (\mathrm{tr}\,\rho W_{1}-\zeta^{\ast}(g_{1}%
)(W_{1}))+(\mathrm{tr}\,\sigma W_{2}-\zeta^{\ast}(g_{2})(W_{2}))+\zeta^{\ast
}(g_{f})(\mathbf{1})\right\} \\
&  =\left\{
\begin{array}
[c]{cc}%
\zeta^{\ast}\left(  g_{f}\right)  (\mathbf{1}), & \text{if }\zeta^{\ast
}\left(  g_{1}\right)  \left(  W\right)  =\mathrm{tr}\,\rho W\text{ and }%
\zeta^{\ast}\left(  g_{2}\right)  \left(  W\right)  =\mathrm{tr}\,\sigma
W,\,\\
\infty, & \text{otherwise.}%
\end{array}
\right.
\end{align*}

Observe $h\rightarrow\zeta^{\ast}(h)(W)$ is a bounded functional on
$\mathcal{C}$. Therefore, by Riesz-Markov representation theorem,
\[
\zeta^{\ast}(h)(W)=\int h(s,t)\mathrm{d}\nu_{W},
\]
where $\nu_{W}$ is a regular measure over the Borel sets of $\left[
0,1\right]  ^{\times2}$. By $\zeta^{\ast}\left(  \chi\left(  B\right)
\right)  \left(  W\right)  =\nu_{W}\left(  B\right)  $, where $\chi\left(
\cdot\right)  $ is the indicator function,
\begin{equation}
|\nu_{W}(B)|\leq\Vert W\Vert\Vert\zeta^{\ast}(\chi(B))(\cdot)\Vert=\Vert
W\Vert\Vert\zeta^{\ast}(\chi(B))(\mathbf{1})\Vert=\Vert W\Vert\nu_{\mathbf{1}%
}(B). \label{|nu|<}%
\end{equation}
Therefore, $\nu_{W}$ is absolutely continuous relative to $\nu_{\mathbf{1}}%
$.Thus $\eta_{s,t}\left(  W\right)  \colon=\frac{\mathrm{d}\nu_{W}}%
{\mathrm{d}\nu_{\mathbf{1}}}$ exists, and
\[
\zeta^{\ast}(h)(W)=\int h(s,t)\eta_{s,t}(W)\mathrm{d}\,\nu_{\mathbf{1}}.
\]
Since $W\rightarrow\zeta^{\ast}(h)(W)$ is linear and positive, so is
$W\rightarrow$ $\eta_{s,t}(W)$, $\nu_{\mathbf{1}}$-a.e. (\ref{|eta|=1-2})
follows from (\ref{|nu|<}). Therefore, rewriting $F_{\ast}$ using $\eta_{s,t}$
and $\nu\colon=\nu_{\mathbf{1}}$, we have the LHS of (\ref{dual-pre-pre}).

Also, $G(\cdot)$ is convex, and $\vec{W}_{0}\colon=(w_{1,0}\mathbf{1}%
,w_{2,0}\mathbf{1})$, where $(w_{1,0},w_{2,0})$ is a relative interior point
of $\mathcal{W}_{f}$, satisfies $G(\vec{W}_{0})<0$. Finally,
\begin{align*}
\eta_{s,t}^{0}(W)  &  :=\left\{
\begin{array}
[c]{cc}%
\mathrm{tr}\,\rho W, & \text{if }(s,t)=(1,0),\\
\mathrm{tr}\,\sigma W, & \text{if }(s,t)=(0,1),\\
0, & \text{otherwise}%
\end{array}
\right.  ,\,\,\\
\nu^{0}(\{(1,0)\})  &  =\nu^{0}(\{(1,0)\}):=1,\,\nu^{0}([0.1]^{\times
2}\backslash\{(0,1),(1,0)\}):=0,\,
\end{align*}
satisfies (\ref{eta-constraint}) and $\int g_{f}(s,t)\eta_{s,t}^{0}%
(\mathbf{1})\mathrm{d}\,\nu^{0}=g_{f}(1,0)+g_{f}(0,1)$ is finite. Thus by
(\ref{weak-dual}), the RHS of (\ref{dual-pre-pre}) is finite. Therefore, we
can apply Proposition\thinspace\ref{prop:duality}, and the assertion is proved.
\end{proof}

\begin{theorem}
Suppose $\mathcal{H}$ a separable Hilbert space, (FC) is satisfied, and
$\hat{f}(0)<\infty$ and $f(0)<\infty$.Then, (i) (\ref{Dmax-dual}) holds if
$W_{1}$ and $W_{2}$ ranges over $\mathcal{B}_{sa}$. (ii) $\inf$ in
(\ref{reverse-test-3}) can be replaced by $\min$. (iii) $\mathrm{D}_{f}^{\max
}$ is lower semicontinuous.
\end{theorem}

\begin{proof}
We use Lemma\thinspace\ref{lem:dual}, and rewrite $\eta_{s,t}$ using $Z(s,t)$.
Then, (i) and (ii) will be simultaneously proved. Since (i) means
$\mathrm{D}_{f}^{\max}$ is the pointwise supremum of linear functionals, (iii)
will follow.

Without loss of generality, one may suppose $f(r)\geq0$, $\forall$ $r\geq0$,
or equivalently, $g_{f}$ is positive. To see this, choose $a$ and $b$ so that
$f_{1}(r):=f(r)-ar-b\geq0$. Then, $g_{f_{1}}(s,t)=g_{f}(s,t)-as-bt\geq0$. If
$(W_{1},W_{2})\in\mathcal{W}_{f_{1}}^{\max}$, $(W_{1}+a,W_{2}+b)\in
\mathcal{W}_{f}^{\max}$, and $\mathrm{D}_{f}^{\max}(\rho\Vert\sigma
)=\mathrm{D}_{f_{1}}^{\max}(\rho\Vert\sigma)+a\mathrm{tr}\,\rho+b\mathrm{tr}%
\,\sigma$. Thus, $\mathrm{D}_{f}^{\max}$ satisfies (i)-(iii) of the present
theorem iff $\mathrm{D}_{f_{1}}^{\max}$ satisfy those.

Since $\eta_{s,t}$ is a bounded linear functional on $\mathcal{B}_{sa}$, there
is $Z(s,t)\in\mathcal{B}_{1,sa}$ with $\mathrm{tr}\,Z(s,t)W=\eta_{s,t}\left(
W\right)  $, for all for any $W$ with finite rank. Then by $\zeta\geq0$ and
(\ref{|eta|=1-2}),
\begin{equation}
Z(s,t)\geq0,\,\,\mathrm{tr}\,Z(s,t)\leq1,\nu\text{-a.e.}. \label{trZ<1}%
\end{equation}
Also,
\begin{equation}
\mathrm{tr}\,Z(s,t)W\leq\eta_{s,t}\left(  W\right)  ,\,W\in\mathcal{B}%
_{sa},W\geq0,\nu\text{-a.e.}. \label{Z<eta}%
\end{equation}

Therefore, since $g_{f}\geq0$ without loss of generality,
\begin{equation}
\int g_{f}(s,t)\eta_{s,t}(\mathbf{1})\mathrm{d}\nu\geq\int g_{f}%
(s,t)\mathrm{tr}\,Z(s,t)\mathrm{d}\nu, \label{zeta>Z}%
\end{equation}
thus replacement of $\eta_{s,t}$ by $W\rightarrow\mathrm{tr}\,Z(s,t)W$ only
improve the value of optimized function.

Next we show that (\ref{eta-constraint}) leads to (\ref{reverse-test-infty}).
Suppose $W\geq0$, and let $\{W^{(k)}\}$ be the sequence of positive finite
rank operators such that $W^{(k)}=\pi_{k}W\pi_{k}$, where $\pi_{k}$ is the
projector onto $k$ -\thinspace dimensional subspace. Then $0\leq W^{(k)}\leq
W$ and as $k\rightarrow\infty$, for , \
\[
s\mathrm{tr}\,Z(s,t)W^{(k)}\nearrow s\mathrm{tr}\,Z(s,t)W\text{, }%
\nu\,\text{-a.e.\thinspace}.
\]
Since the function $(s,t)\rightarrow s\mathrm{tr}\,Z(s,t)W$ is $\nu
$-integrable, by monotone convergence theorem,
\begin{align*}
\mathrm{tr}\,\rho W\underset{(a)}{=}\lim_{k\rightarrow\infty}\mathrm{tr}\,\rho
W^{(k)}  &  =\lim_{k\rightarrow\infty}\int s\eta_{s,t}(W^{(k)})\mathrm{d}\nu\\
&  \underset{(b)}{=}\lim_{k\rightarrow\infty}\int s\mathrm{tr}\,Z(s,t)W^{(k)}%
\mathrm{d}\nu\\
&  =\int\lim_{k\rightarrow\infty}s\mathrm{tr}\,Z(s,t)W^{(k)}\mathrm{d}\nu=\int
s\mathrm{tr}\,Z(s,t)W\mathrm{d}\nu.
\end{align*}
Here, $(a)$ holds since $\rho$ is trace crass, and $(b)$ holds since $W^{(k)}$
is of finite rank. When $W$ is not positive, decomposing it into its positive
and negative part, we obtain the identity. Thus, (\ref{reverse-test-infty}) is satisfied.

Finally, due to (\ref{trZ<1}), $Z(s,t)$ can be normalized to satisfy
(\ref{trZ=1}).
\end{proof}

\section{Discussions}

We had introduced the maximal $f$-\thinspace divergence as the solution to an
optimization problem, reverse test, and shown its closed formula in some
important cases. Next step is to consider asymptotic version of the problem,
in the hope that this close the gap between the maximum and minimum quantum
divergence. The present author's long standing project is to characterize all
the possible quantum $f$-\thinspace divergence, as \cite{Petz} had
characterized all the quantum Fisher information.

\section*{Appendix Matrix analysis}

\begin{proposition}
\label{prop:convex-cfc}(Theorem\thinspace V.2.3 of \cite{Bhatia})Let $f$ be a
continuous function on $[0,\infty)$ . Then, if $f$ is operator convex and
$f(0)\leq0$, for any positive operator $X$ and an operator $C$ such that
$\left\Vert C\right\Vert \leq1$, $f\left(  C^{\dagger}XC\right)  \leq
C^{\dagger}f\left(  X\right)  C$.
\end{proposition}

\begin{proposition}
\label{prop:jensen}((2.43) of \cite{Bhatia-2}) Let $f$ be a operator convex
function defined on $[0,\infty)$. Let $\Lambda^{\dagger}$ be a unital positive
map. Then
\[
f\left(  \Lambda^{\dagger}\left(  A\right)  \right)  \leq\Lambda^{\dagger
}\left(  f\left(  A\right)  \right)
\]
holds for any $A\geq0$. \ 
\end{proposition}

\begin{proposition}
\label{prop:lowner}(Proposition\thinspace8.4 of \cite{HiaiMosonyiPetzBeny})Let
$f$ be a continuous operator convex function on $[0,\infty)$. Then, if
$\hat{f}(0)<\infty$, there is a real number $a$ and a positive Borel measure
$\mu$ such that
\[
f(r)=f(0)+\hat{f}(0)r+\int_{\left(  0,\infty\right)  }\psi_{\lambda
}(r)\mathrm{d}\mu\left(  t\right)  ,\,\,\,\,\psi_{\lambda}(r):=-\frac
{r}{r+\lambda},
\]
and $\int_{\left(  0,\infty\right)  }\frac{\mathrm{d}\mu\left(  \lambda
\right)  }{1+\lambda}<\infty$. Since $\psi_{\lambda}$ is operator monotone
decreasing, this means that $f(r)$ is sum of linear function and operator
monotone decreasing function.
\end{proposition}

\begin{proposition}
\label{prop:f-finite}(Lemma\thinspace5.2 of \cite{HiaiMosonyiPetzBeny}) If $f$
is a complex-valued function on finitely many points $\left\{  r_{i};i\in
I\right\}  \subset\lbrack0,\infty)$, then for any pairwise different positive
numbers $\left\{  \lambda_{i};i\in I\right\}  $ there exist complex numbers
$\left\{  c_{i};i\in I\right\}  $ such that $f(r_{i})=\sum_{i\in I}\frac
{c_{i}}{r_{i}+\lambda_{i}}$ , $i\in I$.
\end{proposition}

\begin{proposition}
\label{prop:block-positive}(Exercise\thinspace1.3.5 of \cite{Bhatia-2}) Let
$X$, $Y$ be a positive definite matrices. Then,
\begin{equation}
\left[
\begin{array}
[c]{cc}%
X & C\\
C^{\dagger} & Y
\end{array}
\right]  \geq0 \label{block}%
\end{equation}
implies
\begin{equation}
X\geq CY^{-1}C^{\dagger},\,\,\,Y\geq C^{\dagger}X^{-1}C. \label{X>Y}%
\end{equation}

\end{proposition}


\begin{thebibliography}{99}                                                                                               %


\bibitem {AmariNagaoka}Amari, S., Nagaoka,H.: Methods of Information Geometry.
AMS (2001)

\bibitem {Bhatia}Bhatia, R.: Matrix Analysis. Springer, Berlin (1996)

\bibitem {Bhatia-2}Bhatia,R.: Positive Definite Matrices. Princeton (2007)

\bibitem {Belavkin}Belavkin,V. P.:On Entangled Quantum Capacity. In: Quantum
Communication, Computing, and Measurement 3.pp.325-333. Kluwer, Boston (2001)

\bibitem {Chefles}Chefles, A.:Deterministic quantum state transformations.
Phys. Lett A 270, 14 (2000)

\bibitem {Ebadian}Ebadian, A., Nikoufar, I., and Gordjic,M.: Perspectives of
matrix convex functions. Proc. Natl Acad. Sci. USA, 108(18), 7313--7314 (2011)

\bibitem {Effros}Effros, E., and Hansen, F.,: Non-commutative perspectives,
Ann. Funct. Anal. Volume 5, Number 2, 74-79 (2014)

\bibitem {Luenberger}Luenberger, D. G.:Optimization by vector space methods.
Wiley, New York (1969)

\bibitem {HammersleyBelavkin}Hammersley,S. J.,Belavkin, V. P.:Information
Divergence for Quantum Channels, Infinite Dimensional Analysis. In: Quantum
Information and Computing, Quantum Probability and White Noise Analysis,VXIX,
pp.149-166, World Scientific, Singapore(2006)

\bibitem {Hayashi}Hayashi,M.:Characterization of Several Kinds of Quantum
Analogues of Relative Entropy. Quantum Information and Computation, Vol. 6,
583-596 (2006)

\bibitem {HiaiMosonyi}Hiai,F., Petz,D.: Different quantum f-divergences and
the reversibility of quantum operations, arXiv:math-ph/1604.03089 (2006)

\bibitem {HiaiMosonyiPetzBeny}Hiai, F., Mosonyi, M., Petz D., and Beny,
C.:Quantum f- divergences and error corrections. Rev. Math. Phys. 23, 691--747 (2011)

\bibitem {HiaiPetz}Hiai,F., Petz,D.: The proper formula for relative entropy
and its asymptotics in quantum probability. Comm. Math. Phys. 143, 99-114 (1991)

\bibitem {Holevo}Holevo, A.S.:Probabilistic and Statistical Aspects of Quantum
Theory, North-Holland, Amsterdam, (1982)(in Russian, 1980)

\bibitem {Matsumoto:dr}Matsumoto, K.: A Geometrical Approach to Quantum
Estimation Theory, doctoral dissertation, University of Tokyo (1998)

\bibitem {Matsumoto:05}Matsumoto, K.: Reverse estimation theory,
Complementality between RLD and SLD, and monotone distances.
arXiv:quant-ph/0511170 (2005)

\bibitem {Matsumoto}Matsumoto, K.: Reverse test and quantum analogue of
classical Fidelity and generalized Fidelity, arXiv:quant-ph/1006.0302 (2010)

\bibitem {Matsumoto:14}Matsumoto, K.: On maximization of measured
$f$-divergence between a given pair of quantum states, arXiv:1412.3676 (2014)

\bibitem {Jencova:03}Jencova, A.:Affine connections, duality and divergences
for a von Neumann algebra. arXiv:math-ph/0311004 (2003)

\bibitem {Jencova}Jencova, A.:Reversibility conditions for quantum operations.
Rev. Math. Phys. 24 1250016(2012)

\bibitem {Parthasarathy}Parthasarathy, K.:Probability and Measures on Metric
Spaces. Academic Press(1967)

\bibitem {Petz}Petz,D.:Monotone Metrics on Matrix Spaces:Linear Algebra and
its Applications, 224, 81-96 (1996)

\bibitem {Rockafellar}Rockafellar,R.T.:Convex Analysis. Princeton(1970)

\bibitem {Ryan}Ryan,R.A.:Introduction to tensor products of Banach spaces,
Springer, Berlin(2002)

\bibitem {Strasser}Strasser, H.:Mathematical Theory of
Statistics---Statistical Experiments and Asymptotic Decision Theory. Walter de
Gruyter, Berlin(1985).

\bibitem {Uhlmann}Uhlmann, A.:Eine Bemerkung uber vollstandig positive
Abbildungen von Dichteopera-toren. Wiss. Z. KMU Leipzig, Math.-Naturwiss. R.
34(6), 580-582 (1985).
\end{thebibliography}
\end{document}